%% file: erm-ts-reg-rev-2.tex
\documentclass[a4paper,12pt]{article}

\usepackage{amsmath, amssymb, amsthm}
\usepackage{bbm, bm, bold-extra}
\usepackage{dsfont}
\usepackage{fncylab}
\usepackage{enumerate}
\usepackage[T1]{fontenc}
\usepackage{graphicx}
\usepackage{hyperref}
\usepackage[utf8]{inputenc}
\usepackage{mathrsfs, mathtools}
\usepackage[authoryear,longnamesfirst,round]{natbib}
\usepackage{setspace}
\usepackage{xcolor}
\usepackage[left = 2.5cm, right = 2.5cm, top = 3cm, bottom = 2.5cm]{geometry}

\newtheorem{prop}{Proposition}
\newtheorem*{prop-nonumber}{Proposition}
\newtheorem{lemma}{Lemma}
\newtheorem{thm}{Theorem}
\newtheorem*{thm-nonumber}{Theorem}

\newtheorem*{asmredef*}{\assumptionnumber}
\providecommand{\assumptionnumber}{}
\makeatletter
\newenvironment{asmredef}[2]
 {%
  \renewcommand{\assumptionnumber}{A.#1* (#2)}%
  \begin{asmredef*}%
  \protected@edef\@currentlabel{A.#1*}%
 }
 {%
  \end{asmredef*}
 }
\makeatother

\newtheorem*{asmredef2*}{\assumptionnumber}
\providecommand{\assumptionnumber}{}
\makeatletter
\newenvironment{asmredef2}[2]
 {%
  \renewcommand{\assumptionnumber}{A.#1** (#2)}%
  \begin{asmredef2*}%
  \protected@edef\@currentlabel{A.#1**}%
 }
 {%
  \end{asmredef2*}
 }
\makeatother

\newtheoremstyle{mytheorem}
  {\topsep}
  {\topsep}
  {\itshape}
  {0pt}
  {\bfseries}
  {.}
  { }
  {\thmname{#1}.\thmnumber{#2}\thmnote{ (#3)}}
\theoremstyle{mytheorem}

\newtheorem{asm}{A}
\labelformat{asm}{A.\theasm}

\title {\large{\textsc{Performance of Empirical Risk Minimization \\ for Linear Regression with Dependent Data}}}
\author{\normalsize{Christian Brownlees}$^{\dag,*}$ \and \normalsize{Gu\dh mundur Stef\'an Gu\dh mundsson}$^{\ddag}$} 

\long\def\symbolfootnote[#1]#2{\begingroup\def\thefootnote{\fnsymbol{footnote}}\footnote[#1]{#2}\endgroup}
\begin{document}    
      
\maketitle

\begin{abstract}

This paper establishes bounds on the performance of empirical risk minimization for large-dimensional linear regression. 
We generalize existing results by allowing the data to be dependent and heavy-tailed.
The analysis covers both the cases of identically and heterogeneously distributed observations.
Our analysis is nonparametric in the sense that the relationship between the regressand and the regressors is not specified. 
The main results of this paper show that the empirical risk minimizer achieves the optimal performance (up to a logarithmic factor) in a dependent data setting. 



\end{abstract}

\symbolfootnote[0]{\\
\noindent
$^{\dag}$ Department of Economics and Business, Universitat Pompeu Fabra and Barcelona GSE;\\
e-mail: \texttt{christian.brownlees@upf.edu}.\\
$^{\ddag}$ Department of Economics and Business Economics, Aarhus University; \\
e-mail: \texttt{gsgudmundsson@econ.au.dk}.\\
$^*$ Corresponding author. \\ 
We have benefited from discussions with Liudas Giraitis, Emmanuel Guerre, Petra Laketa, Gabor Lugosi, Stanislav Nagy, Jordi Llorens-Terrazas, Yaping Wang and Geert Mesters as well as seminar participants at the Granger Center, Nottingham University; School of Economics and Finance, Queen Mary University of London.
Christian Brownlees acknowledges support from the Spanish Ministry of Science and Technology (Grant MTM2012-37195)
and the Spanish Ministry of Economy and Competitiveness through the Severo Ochoa Programme for Centres of Excellence in R\&D (SEV-2011-0075).
Gu\dh mundur Stef\'an Gu\dh mundsson acknowledges financial support from the 
Danish National Research Foundation (DNRF Chair grant number DNRF154).
}

\newpage
\doublespacing

\section{Introduction}

Let $ \mathcal D = \{ (Y_t, \bm X_t')' \}_{t=1}^T $ be a sequence of dependent random vectors taking values in $\mathcal Y \times \mathcal X$ with $\mathcal Y \subset \mathbb R$ and $\mathcal X \subset \mathbb R^p$.
The $p$-dimensional vector $\bm X_t=(X_{1\,t},\ldots,X_{p\,t})'$ is used to predict the variable $Y_t$ through the class of linear forecasts given by
 \begin{equation}\label{eqn:for}
	f_{\bm \theta\,t} = \theta_1 X_{1\,t} + \ldots + \theta_p X_{p\,t} ~,
\end{equation}
where $(\theta_1,\ldots,\theta_p)' = \bm \theta \in \mathbb R^p$.
As is customary in learning theory, the relation between the regressand $Y_t$ and the regressors $X_{1\,t},\ldots,X_{p\,t}$ is not specified, 
and \eqref{eqn:for} should be interpreted as a class of prediction rules indexed by $\bm \theta \in \mathbb R^p$.

A prediction rule is to be chosen from the data.
The precision of a prediction rule is measured by its average risk defined as
\[
	R( \bm \theta ) = \mathbb E \left[  {1 \over T} \sum_{t=1}^{T} ( Y_{t} - f_{\bm \theta\,t} )^2 \right] ~.
\]
Thus, a natural strategy for choosing a prediction rule from the data consists in minimizing the empirical risk.
The empirical risk minimizer (ERM) is defined as
\begin{equation}\label{eqn:erm}
	\hat{\bm \theta} \in \arg \min_{\bm \theta \in \mathbb R^p } R_T(\bm \theta) ~, 
	\text{ where }
	R_T(\bm \theta) = {1 \over T} \sum_{t=1}^{T} ( Y_{t} - f_{\bm \theta\,t} )^2  ~.
\end{equation}
If more than one prediction rule achieves the minimum we may pick one arbitrarily.
Clearly, the ERM in \eqref{eqn:erm} corresponds to the classic least squares estimator.
We sometimes denote $\hat{ \bm \theta}$ as $\hat{ \bm \theta}(\mathcal D)$ to emphasize that the ERM is a function of the data $\mathcal D$.
The problem we have described so far is known as linear regression in statistics and econometrics whereas in learning theory it is known as linear aggregation \citep{Nemirovski:2000}.

The accuracy of the ERM is measured by its conditional average risk defined as
\begin{equation}\label{eqn:risk:erm}
	R( \hat{ \bm \theta} )  
	= \mathbb E \left[ \left. {1 \over T} \sum_{t=1}^{T} ( Y_{t} - \hat f_{t} )^2 \right| \hat{ \bm \theta} = \hat{ \bm\theta} ( \mathcal D' ) \right] ~,
\end{equation}
where $\hat f_t = \hat \theta_1 X_{1\,t} + \ldots + \hat \theta_p X_{p\,t}$ and $\mathcal D'$ denotes an independent copy of the data $ \mathcal D $.
The performance measure in \eqref{eqn:risk:erm} can be interpreted as the risk of the ERM obtained from the ``training data'' $\mathcal D'$ over the ``validation data'' $\mathcal D$.
This performance measure allows us to keep our analysis close to the bulk of contributions in the learning theory literature (which typically focus on the analysis of i.i.d.~data) and facilitates comparisons.
We also consider as an alternative accuracy measure the conditional out-of-sample average risk of the ERM, which is more attractive for time series applications.
The alternative measure leads a to similar result at the expense of introducing additional notation.

The main objective of this paper is to obtain a bound on the performance of the ERM relative to the optimal risk that can be achieved within the given class of prediction rules.
We aim to establish a bound $B_T(p)$ such that $B_T(p) \rightarrow 0$ as $T \rightarrow \infty$ for which
\begin{equation}\label{eqn:risk:regret}
	R( \hat{ \bm \theta } ) \leq \inf_{\bm \theta \in \mathbb R^p} R(\bm \theta) + B_T(p) 
\end{equation}
holds, with high probability, for all (sufficiently large) $T$.
The inequality in \eqref{eqn:risk:regret} is commonly referred to as an \emph{oracle inequality}.
Oracle inequalities such as \eqref{eqn:risk:regret} provide non-asymptotic guarantees on the performance of the ERM.
The inequality in \eqref{eqn:risk:regret} implies that empirical risk minimization achieves asymptotically the best performance that is possible to attain in the class.
We emphasize that in this paper we study the performance of the ERM for large-dimensional linear regression,
meaning that in our analysis we assume that the number of predictors $p$ is not negligible relative to $T$ (in a sense to be spelled out precisely below).
Establishing bounds on the performance of the ERM is a classic problem in learning theory.
There is a fairly extensive literature that has studied this problem in the i.i.d.~setting \citep{AudibertCantoni:2010}.
The literature conveys that the best possible rate for $B_T(p)$ is of the order $p/T$, which is referred to as the optimal rate of linear aggregation \citep{Tsybakov:2003}.

The main contribution of this paper consists in establishing oracle inequalities for the ERM when
the data are dependent and heavy-tailed.
Our analysis covers both the cases of identically and heterogeneously distributed observations (using the jargon of \citet{White:2001}).  
In particular, our main results establish that the ERM achieves the optimal rate of linear aggregation (up to a $\log(T)$ factor) in a dependent data setting.
Our analysis highlights a trade-off between the dependence and moment properties of the data on the one hand, and the number of predictors on the other.
In particular we show that the higher the dependence and the lower the number of moments of the data, the lower the maximum rate of growth allowed for the number of predictors.
We emphasize that our analysis is nonparametric, in the sense that the relationship between the regressand and the regressors is assumed to be unknown.
Lastly, we remark that the performance bound we recover depends transparently on constants that are straightforward to interpret.

Four remarks are in order before we proceed.
{\input{fourth_paragraph}}

{
Second, this paper has a number of connections with the nonparametric literature and, in particular, with nonparametric series methods \citep{Stone:1985,Andrews:1991,Newey:1997,Chen:Shen:1998,Chen:2006,Tsybakov:2014,Belloni:2015}. 
Among these papers we remark that \cite{Chen:Shen:1998} is the only one that considers a non i.i.d.~data setup.
Let $ \{ (Y_t, \bm W_t')' \}_{t=1}^T $ be a strictly stationary sequence of random vectors in $\mathcal Y \times \mathcal W \subset \mathbb R \times \mathbb R^d$.
Then our framework subsumes the problem of estimating the conditional mean of $Y_t$ given $\bm W_t$ on the basis of the approximation given by 
\begin{equation*}
    \mathbb E( Y_t | \bm W_t ) \approx \theta_1 f_1(\bm W_t) + \ldots + \theta_p f_p(\bm W_t)  ~,
\end{equation*}
where $\{ f_i \}$ with $f_i : \mathcal W \rightarrow \mathbb R$ is a collection of functions (e.g.~B-splines) called a dictionary.
We emphasize that, in some sense, our framework is more general since our focus lies on the estimation of the optimal linear prediction rule rather than the conditional mean.
}

Third, the literature on empirical risk minimization and oracle inequalities for dependent data has been rapidly developing in recent years. 
Notable contributions in this area include the works of \citet{JiangTanner:2010}, \citet{Fan:Liao:Mincheva:2011}, \citet{CanerKnight:2013}, \citet{LiaoPhillips:2015} and \citet{MiaoPhillipsSu:2020}.
We remark that one of the challenges of this literature is that it is not straightforward to apply the theoretical machinery used in learning theory in a dependent data setting.
In fact, as forcefully argued in \citet{Mendelson:2015}, several of the standard results on empirical risk minimization used in learning theory assume i.i.d.~bounded data and cannot be extended beyond this setup.
In this work we rely on a proof strategy based on the so-called \emph{small-ball} method developed by Shahar Mendelson and Guillaume Lecu\'e \citep{Mendelson:2015,LecueMendelson:2016}. 
The small-ball method allows us to establish sharp bounds on the performance of the ERM under fairly weak moment and dependence assumptions.

Fourth, our analysis aims to provide large-dimensional analogues of some of the classic results of \citet{White:2001} for fixed-dimensional linear regression with dependent data.
We shall point out the differences between those results and the ones established here. 

This paper is related to various strands of the literature.
First, it is related to the literature on empirical risk minimization for linear aggregation,
which includes \citet{BirgeMassart:1998}, \citet{BuneaTsybakovWegkamp:2007}, \citet{AudibertCantoni:2011} and \citet{LecueMendelson:2016}.
Second, it is related to the literature on empirical risk minimization for heavy-tailed data,
which includes \citet{AudibertCantoni:2011} and \citet{BrownleesJolyLugosi:2015}.
Third, it is related to the literature on empirical risk minimization for dependent data. In particular this contribution is close to \citet{JiangTanner:2010}.
Fourth, this paper is related to the vast literature on nonparametric estimation and nonparametric series methods, which includes \citet{Chen:2006} and \citet{Belloni:2015}.
\citet{Li:Racine:2006} contains a number of important results and references to this literature.
Fifth, it is related to the literature on the small-ball method,
which includes \citet{Mendelson:2018}, \citet{LecueMendelson:2017} and \citet{LecueMendelson:2018}.
Sixth, it is related to the vast literature on machine learning and large-dimensional modeling,
which includes (in econometrics) \citet{Kock:Callot:2015}, \citet{Medeiros:Mendes:2016}, \citet{Garcia:2017} and \citet{Babii:Ghysels:Striaukas:2021}.
\cite{Hastie:2001} and \cite{wainwright_2019} contain a number of important results and references to this literature.

The rest of the paper is structured as follows.
Section \ref{sec:asm} contains preliminaries, additional notation and assumptions.
Section \ref{sec:hetero} contains an oracle inequality for linear regression with heterogeneously distributed observations.
Section \ref{sec:homo} contains an analogous result for identically distributed observations.
Section \ref{sec:ext} contains extensions of the baseline results.
Concluding remarks follow in Section \ref{sec:end}.
All proofs are in the Appendix.

\section{Notation, Preliminaries and Assumptions} \label{sec:asm}

We introduce the notation used in the remainder of the paper.
For a generic vector $\bm x \in \mathbb R^d$ we define $\| \bm x \|_{r}$ as $[ \sum_{i=1}^d |x_i|^r ]^{1/r}$ for $ 1 \leq r < \infty$ and $\max_{i=1,\ldots,d} |x_i| $ for $r=\infty$.
For a generic random variable $X \in \mathbb R$ we define $\| X \|_{L_r}$ as $[\mathbb E ( |X|^r ) ]^{1/r}$ for $ 1 \leq r < \infty$ and $\inf \{ a : \mathbb P( |X|> a ) = 0 \}$ for $r=\infty$.
For a positive semi-definite matrix $\mathbf M$ we use $\mathbf M^{ {1 \over 2} }$ to denote the
positive semi-definite square root matrix of $\mathbf M$
and $\mathbf M^{ -{1 \over 2} }$ to denote the generalized-inverse of $\mathbf M^{1\over 2}$.

In this section we establish a preliminary result and introduce the main assumptions required in our analysis.
All results and assumptions are stated for the case of heterogeneously distributed observations.
Clearly, these simplify in a straightforward manner if the observations are identically distributed. 

We begin by establishing the existence of the optimal prediction rule, that is the \emph{oracle}.
Lemma \ref{lemma:opt} states that there exists an optimal $\bm \theta^*$ that satisfies a Pythagorean-type identity.
We remark that the assumptions of Lemma \ref{lemma:opt} are fairly weak and, in particular, weaker than what we require for the analysis of the ERM. 

\begin{lemma}\label{lemma:opt}
Let $\{Y_t\}_{t=1}^T$ satisfy $\sup_{1\leq t\leq T} \| Y_{t} \|_{L_2} < \infty $ and $\sup_{1\leq i\leq p} \sup_{1\leq t\leq T}\| X_{i\,t} \|_{L_2} < \infty $. 

Then 
\begin{enumerate}[(i)]
\item  there exists a $\bm \theta^* \in \mathbb R^p$ such that  
\[
	\bm \theta^* \in \arg\min_{\bm \theta \in \mathbb R^p} R(\bm \theta) ;
\] 
\item $\bm \theta^*$ is such that for any $\bm \theta \in \mathbb R^p$ it holds that 
\[
	{1 \over T} \sum_{t=1}^{T} \| Y_t - f^*_{t} \|^2_{L_2} + {1 \over T} \sum_{t=1}^{T} \| f^*_{t} - f_{\bm \theta\,t} \|^2_{L_2} = {1 \over T} \sum_{t=1}^{T}  \| Y_t - f_{\bm \theta\,t} \|^2_{L_2} ~,
\]
where $f^*_{t} = f_{\bm \theta^*\,t}$;
\item if $ \sum_{t=1}^T \mathbb E \bm X_t \bm X_t'$ is positive definite then $\bm \theta^*$ is unique.
\end{enumerate}
\end{lemma}

Next, we lay out the assumptions we require to establish the properties of the ERM.

\begin{asm}[Moments]\label{asm:moments}
The sequences $\{Y_t\}_{t=1}^T$, $\{ \bm X_{t} \}_{t=1}^T$, $\{f^*_t\}_{t=1}^T$ satisfy 
$\sup_{1\leq t\leq T}  \| Y_t \|_{L_{r_m}} \leq K_m $,
$\sup_{1 \leq i \leq p} \sup_{1\leq t\leq T}  \| X_{i\,t} \|_{L_{r_m}} \leq K_m $ and
$\sup_{1 \leq i \leq p} \sup_{1\leq t\leq T} \| (Y_{t}-f^*_t) X_{i\,t} \|_{L_{r_m}} \leq K_m $, 
for some $K_m \geq 1$ and $r_m > 2$.
\end{asm}

Assumption \ref{asm:moments} states that the regressand, predictors, 
and the product of the predictors and the forecast error of the optimal prediction rule 
have a number of moments strictly larger than two.
The assumption also states that the $r_m$-th moments are bounded by a constant $K_m$ uniformly in $t$. 
A few comments are in order.
First, this moment assumption is formulated as in \citet[Chapter 3]{White:2001} in the analysis of linear regression with heterogeneous data.
Alternatively, we may state this assumption for the forecast error of the optimal prediction rule and the predictors separately and require at least four moments to exist and to be uniformly bounded.
{Second, we assume $K_m \geq 1$ to obtain simpler expressions of some of the constants that appear in our analysis.
Note that this is without loss of generality.
}
Lastly, we emphasize that this assumption is weaker than what is assumed in a number of contributions on oracle inequalities for dependent data for large dimensional models 
such as \citet{JiangTanner:2010}, \citet{Fan:Liao:Mincheva:2011} and \citet{Kock:Callot:2015} which assume that all moments exist.
We remark that assuming that the moments are uniformly bounded is fairly standard in the analysis of regression models with heterogeneous dependent data and that 
requiring more than two moments to exist is also required to establish consistency of the least
squares estimator for fixed-dimensional linear regression \citep[Chapter 3]{White:2001}.

\begin{asm}[Dependence]\label{asm:mixing}
Let $\mathcal{F}_{-\infty}^s$ and $\mathcal{F}_{s+l}^{\infty}$ be the $\sigma$-algebras generated by $\lbrace (Y_t, \bm X_t')': -\infty \leq t \leq s\rbrace$ and $\lbrace (Y_t, \bm X_t')': s + l \leq t \leq \infty\rbrace$ respectively and 
define the $\alpha$-mixing coefficients 
\begin{equation*}
\alpha(l) = \sup_s \sup_{A \in \mathcal{F}_{-\infty}^s, B \in \mathcal{F}_{s+l}^{\infty}} \left| {\mathbb P \left(A \cap B \right) - \mathbb P \left(A\right) \mathbb P \left(B\right) } \right|.
\end{equation*}
The $\alpha$-mixing coefficients satisfy $ \alpha(l) \leq \exp( -K_\alpha l^{r_\alpha} ) $ for some $K_\alpha>0$ and $r_\alpha>0$. 
\end{asm}

{Assumption \ref{asm:mixing} states that the sequence $\{ (Y_t, \bm X_t')'\}_{t=1}^T$ is strongly mixing with geometrically decaying mixing coefficients.
The definition of the mixing coefficients is as in \citet[Definition 3.42]{White:2001} and does not hinge on the data generating process being stationarity.
See also \citet{SuWhite:2010} for the analysis of $\alpha$-mixing processes that are not required to be stationary.
Note that while this is a stronger assumption than what is required by
classical results for consistency and asymptotic normality for the (finite-dimensional) linear regression model that rely on polynomial $\alpha$-mixing \citep[Chapter 3]{White:2001}, 
geometric $\alpha$-mixing is commonly used in the analysis of large dimensional time series models \citep{JiangTanner:2010,Fan:Liao:Mincheva:2011,Kock:Callot:2015}.
Moreover, geometric $\alpha$-mixing is satisfied by many commonly encountered processes such as ARMA and GARCH \citep{Meitz:Saikkonen:2008}.}

\begin{asm}[Number of Predictors]\label{asm:nparams}
The number of predictors satisfies $p = \lfloor K_p T^{r_p} \rfloor$ for some $K_p > 0$ and $0 \leq r_p < {r_\alpha \over r_\alpha+1} \wedge {r_m-2 \over 2} $.
\end{asm}

Assumption \ref{asm:nparams} states that the number of predictors is a function of $T$.
This assumption allows the number of predictors to be constant or to grow sublinearly in $T$.
Importantly, the bound on the rate of growth of the number of predictors $p$ depends on 
the number of moments and the amount of dependence of the data. 
The more moments and the less dependence, the higher the maximum rate of growth of the number of predictors.
If the data have at least four moments, then the number of predictors is only constrained by the amount of dependence in the data.

\begin{asm}[Eigenvalues]\label{asm:eigen}
Define $ \bm \Sigma_t = \mathbb E ( \bm X_t \bm X_t' )$ and let $\lambda_{\min}\left( \bm \Sigma_t \right)$ 
and $\lambda_{\max}\left( \bm \Sigma_t \right)$ be the smallest and 
largest eigenvalue of $\bm \Sigma_t$ respectively. 
Then the sequence $\{ \bm \Sigma_t \}_{t=1}^T$ satisfies 
(i) $\underline{\lambda} \leq \inf_{1\leq t \leq T} \lambda_{\min}\left( \bm \Sigma_t \right)$ 
for some $ 0<\underline{\lambda} $ 
and (ii) $ \sup_{1\leq t \leq T}\lambda_{\max}\left( \bm \Sigma_t \right) \leq \overline{\lambda} $ 
for some $ 0< \overline{\lambda} < \infty$.
\end{asm}

Assumption \ref{asm:eigen} states that the eigenvalues of the covariance matrix of the predictors are bounded from above and bounded away from zero uniformly in $t$.
The assumption that the smallest eigenvalue is bounded away from zero is fairly standard \citep{Newey:1997}.
Notice that $\bm \theta^*$ is unique when \ref{asm:eigen}$(i)$ holds, by Lemma \ref{lemma:opt}$(iii)$. 
Assuming that the largest eigenvalue of the covariance matrix of the predictors is bounded above uniformly in $t$ is more restrictive. 
We remark that, as we shall see in detail below, under the additional assumption of identically distributed observations these constraints can be relaxed.
In what follows we shall also use the constant \( K_{\bm \Sigma} = { \overline{\lambda} / \underline{\lambda} } \),
which is an upper bound on the condition number of the matrices $\{ \bm \Sigma_t \}_{t=1}^T$ and measures the maximum degree of collinearity between the predictors.

\begin{asm}[Distribution]\label{asm:distribution}
Consider the sequence of random vectors $\{ \bm Z_t \}_{t=1}^T$ with $\bm Z_t = \bm \Sigma_t^{-{1\over 2}} \bm X_t$.
Then $\sup_{1 \leq t \leq T} \mathbb P( \bm Z_t \in E ) \leq K_{\bm Z} \mathbb P( \bm S \in E ) $ holds 
for some $p$-dimensional spherical random vector $\bm S$, some positive constant $K_{\bm Z}$ and any $E \in \mathcal B(\mathbb R^p)$.
The density of $\bm S$ exists and the marginal densities of the components of $\bm S$ are bounded from above.
\end{asm}

Assumption \ref{asm:distribution} is required to establish upper bounds on the probability of a certain event associated with the vector of predictors $\bm X_t$ in one of the intermediate propositions of our analysis.
The probability of this event boils down to a multiple integral that can be expressed using $n$-spherical coordinates.
The spherical distribution bound in \ref{asm:distribution} makes it easy to compute such an integral after the $n$-spherical coordinates transformation. 
We conjecture that the assumption could be relaxed, however this would be at the expense of more tedious computations. 
That being said, the family of spherical distributions is fairly large and includes the appropriately standardized versions of the multivariate Gaussian, Student $t$, Cauchy and uniform\footnote{To be precise, the multivariate uniform distribution over the sphere.} distributions.\footnote{For more details on the class of spherical distributions we refer to \citet{Fang:Kotz:Ng:1990}.}
Moreover, finite mixtures of spherical distributions are also spherical.
Assumption \ref{asm:distribution} may be interpreted as a generalization of the bounded density assumption typically encountered in the nonparametric literature \citep{Newey:1997,Li:Racine:2006,Hansen:2008}.
Bounded density assumptions are also formulated in \citet{JiangTanner:2010} in the analysis of empirical risk minimization for time series data with bounded support.
Last, we remark that this assumption allows for weaker moment conditions than what is imposed by Assumption \ref{asm:moments}.

\begin{asm}[Identification/Small-ball] \label{asm:smallball}
The sequence $\{ \bm X_t \}_{t=1}^T$ satisfies, for each $t=1,\ldots,T$ and for each $\bm \theta_1, \bm \theta_2 \in \mathbb R^p$,
\[
	\mathbb P \left( | f_{\bm \theta_1\,t} - f_{\bm \theta_2\,t} | \geq \kappa_1 \| f_{\bm \theta_1\,t} - f_{\bm \theta_2\,t} \|_{L_2} \right) \geq \kappa_2 ~,
\]
for some $\kappa_1>0$ and $\kappa_2>0$.
\end{asm}

Assumption \ref{asm:smallball} is the so-called small-ball assumption, and it is stated here as it is formulated in \citet{LecueMendelson:2016}.
This assumption can be interpreted as an identification condition.
If we define $\bm v=(\bm \theta_1-\bm \theta_2)$ then the condition is equivalent to \( \mathbb P \left( | \bm v' \bm X | \geq \kappa_1 \| \bm v' \bm X \|_{L_2} \right) \geq \kappa_2 \),
which can be seen as requiring that the random variable $\bm v' \bm X$ does not have excessive mass
in a neighbourhood around zero. 
We remark that the constants $\kappa_1$ and $\kappa_2$ measure the strength of the identification
in the sense that the larger the value of these constants the stronger the identification condition is.
In Section \ref{sec:ext} we establish alternative identification assumptions that in turn imply Assumption \ref{asm:smallball}.

\section{Dependent Heterogeneously Distributed Observations}\label{sec:hetero}

The ERM performance bound that we derive in this section depends on a constant related to the variance of the gradient of the empirical risk evaluated at the optimal prediction rule (after an appropriate rescaling), that is $\operatorname{Var} \left( {1\over \sqrt{T} } \sum_{t=1}^T (Y_t - f^*_t) \bm X_t \right)$.
As is well known, in the standard large sample analysis of linear regression 
the asymptotic variance of the least squares estimator is typically expressed as a function of the limit of this quantity \citep[Chapter 5]{White:2001}.
In our analysis, the ERM performance depends on an upper bound on the diagonal elements of this quantity that is given by
\[
	K_{\sigma^2} = K_m^2 \left( 1 + 128 {r_m \over r_m -2} \sum_{l=1}^\infty \alpha(l)^{1-{2\over r_m} } \right) ~.
\]
It is possible to make substantially smaller choices of this constant if we make simplifying assumptions on the setup of our analysis.
We explore this in more detail in Section \ref{sec:homo}.

We can now state the main result of this section.

\begin{thm}\label{thm:erm}
	Suppose Assumptions \ref{asm:moments}--\ref{asm:smallball} are satisfied.
	Then, for all $T$ sufficiently large, the empirical risk minimizer defined in \eqref{eqn:erm} satisfies
	\begin{equation}\label{eqn:thm}
	R( \hat{\bm \theta} ) \leq R( \bm \theta^* ) + { K_{\sigma^2} } \, { K^{3}_{\bm \Sigma} \over \underline{\lambda} }  \, \left({ 48 \over \kappa_1^2 \kappa_2  }\right)^2 {p \log (T) \over T} ~,
	\end{equation}
	with probability at least $1-{3 K_p (2K_{m})^{r_m} /( K^{1\over 2}_{\sigma^2} \log (T) )} - o( \log(T)^{-1} )$.
\end{thm}

The theorem establishes that the ERM for large-dimensional linear regression with heterogeneous dependent data achieves the optimal rate of linear aggregation (up to a $\log (T)$ factor). 
{We remark that \ref{asm:nparams} implies that $( p \log (T) )/ T \rightarrow 0$ as $T \rightarrow \infty$, which makes the inequality in the theorem an oracle inequality.}
Note that the bound on the performance of the ERM is proportional to quantities that are associated with a larger asymptotic variability of the least squares estimator.
%
{We remark that Theorem \ref{thm:erm} may be seen as a non-asymptotic version of classic asymptotic results in the series estimation literature, which establish optimality of the nonparametric least squares estimator.
In fact, the convergence rate of $p/T$ (up to a $\log(T)$ factor) is the same as the rate obtained (for instance) in \citet[Theorem 4.1]{Belloni:2015}.}

It is interesting to compare Theorem \ref{thm:erm} with an analogous result for i.i.d.~data. 
The following result \citet[Corollary 1.2]{LecueMendelson:2016} is taken as benchmark.

\begin{thm-nonumber} \label{thm:bench}
	Consider the linear regression model 
	\[ 
	Y_t = \bm X_t' \bm \theta^* + \epsilon_t, \qquad t=1,\ldots,T ~, 
	\]
	where $\{ \bm X_t \}$ and $\{ \epsilon_t \}$ are sequences of i.i.d.~random variables 
    with $\mathbb E( \epsilon_t ) = 0$, $\operatorname{Var}( \epsilon_t ) = \sigma^2$,
    and $\epsilon_t$ is independent of $\bm X_t$.
	Assume that there are constants $\kappa_1$ and $\kappa_2$ such that 
	\[
	\mathbb P \left( | f_{\bm \theta_1\,t} - f_{\bm \theta_2\,t} | \geq \kappa_1 \| f_{\bm \theta_1\,t} - f_{\bm \theta_2\,t} \|_{L_2} \right) \geq \kappa_2 ~,
	\]
	for all $\bm \theta \in \mathbb R^p$.
	Then, for all $T > (400)^2 p / \kappa_2^2$ and $x>0$ we have that the empirical risk minimizer defined in \eqref{eqn:erm} satisfies
	\[
	R( \hat{\bm \theta} ) 
	\leq R( \bm \theta^* ) + { \sigma^2 } \left({ 16 \over \kappa_1^2 \kappa_2 }\right)^2 { {p \over T } } \, x ~,
	\]
	with probability at least $1-\exp( - \kappa_2 T/4 )- (1/x)$.
\end{thm-nonumber}

As is immediate to see, we recover an analogous bound to what is established in \citet{LecueMendelson:2016}.
The constant that appears in our risk bound in \eqref{eqn:thm} is much larger than the one in this benchmark result.
However, we remark that below we obtain a more favourable bound by simplifying the setup of our analysis.
Also, we remark that the result above relies on assuming that the ``true model'' exists. 
\citet{LecueMendelson:2016} also have results that do not depend on such an assumption but rely on stronger assumptions on the prediction errors of the optimal forecast.

We conclude this section with a sketch of the proof.
This is an elegant argument based on \citet{LecueMendelson:2016}.
Define the empirical risk differential for $\bm \theta \in \mathbb R^p$ as
\begin{equation*}
\widehat{\mathcal L}_{\bm \theta} 
= R_T(\bm \theta) - R_T(\bm \theta^*) 
= {1 \over T}\sum_{t=1}^T ( f^*_{t} - f_{\bm \theta\,t} )^2 + {2 \over T} \sum_{t=1}^T (Y_t-f^*_{t} )( f^*_{t} - f_{\bm \theta\,t} ) ~.
\end{equation*}	
The proof is based on showing that if the condition 
\begin{equation}\label{eqn:cond}
{1 \over T} \sum_{t=1}^T \|f^*_t - f_{\bm \theta\,t} \|_{L_2}  > { 48 K^{1\over 2}_{\sigma^2} K_\Sigma \over \underline{\lambda} \kappa_1^2 \kappa_2 } \sqrt{ {p \log (T) \over T}}  ~,
\end{equation}
holds, then we have that 
\begin{equation}\label{eqn:inequality}
{1 \over T}\sum_{t=1}^T ( f^*_{t} - f_{\bm \theta\,t} )^2 > \left| {2 \over T}\sum_{t=1}^T (Y_t-f^*_{t} )( f^*_{t} - f_{\bm \theta\,t} ) \right| ~,
\end{equation}	
with high probability.
This, in turn, implies that for any $\bm \theta$ that satisfies \eqref{eqn:cond} we have $\widehat{\mathcal L}_{{\bm \theta}} > 0$.
Since the empirical risk minimizer $\hat{\bm \theta}$ must satisfy $\widehat{\mathcal L}_{\hat {\bm \theta}} \leq 0$ then, conditional on the same events, we must have that 
\[
	{1 \over T} \sum_{t=1}^T \|f^*_t - \hat f_{t} \|_{L_2} \leq { 48 K^{1\over 2}_{\sigma^2} 
    K_{\bm \Sigma} \over \underline{\lambda}\kappa_1^2 \kappa_2 } \sqrt{ {p \log (T) \over T}}  ~,
\]
which, in turn, implies that 
\[
R(\hat {\bm\theta}) - R(\bm\theta^*) = {1 \over T} \sum_{t=1}^T \| f^*_{t} - \hat f_{t}  \|^2_{L_2}
	\leq { K_{\sigma^2} } \, { K^2_{\bm \Sigma} \over \underline{\lambda} } \,\left({ 48 \over \kappa_1^2 \kappa_2 }\right)^2 \, { {p \log (T) \over T }  } ~,  
\]
by an application of Lemma \ref{lemma:opt}.

The following two propositions are key in establishing that the inequality in \eqref{eqn:inequality} holds with high probability and thus to determine the risk bound in Theorem \ref{thm:erm}.

\begin{prop}\label{lem:1}
	Suppose Assumptions \ref{asm:mixing}, \ref{asm:nparams}, \ref{asm:eigen}, \ref{asm:distribution} and \ref{asm:smallball} are satisfied.
	Then, for all $T$ sufficiently large and any $\bm \theta \in \mathbb R^p$,
	\[
	{1 \over T}\sum_{t=1}^T ( f^*_{t} - f_{\bm \theta\,t} )^2 
	\geq { \kappa_1^2 \kappa_2 \over 2 K_{\bm \Sigma}} {1 \over T} \sum_{t=1}^T \| f^*_t - f_{\bm \theta\,t} \|^2_{L_2} ~,
	\]
	holds with probability at least $1-8 T^{-1} - o(T^{-1})$.
\end{prop}

\begin{prop}\label{lem:2}
	Suppose Assumptions \ref{asm:moments}, \ref{asm:mixing}, \ref{asm:nparams} and \ref{asm:eigen} are satisfied.
	Then, for all $T$ sufficiently large and any $\bm \theta \in \mathbb R^p / \{ \bm \theta^* \}$,
	\begin{equation*}
	\left| {1 \over T} \sum^T_{t=1}(Y_t-f^*_t)(f^*_t- f_{\bm \theta\,t} ) \right| \leq  12 \, \sqrt{ K_{\sigma^2} \over \underline \lambda} \, {1 \over T}\sum^T_{t=1}\| f^*_t- f_{\bm \theta\,t}\|_{L_2} \, \sqrt{ {p \log (T) \over T} }  ~,
	\end{equation*}
	holds with probability at least $1- 3 K_p (2K_{m})^{r_m} / ( K^{1\over 2}_{\sigma^2} \log (T) ) - o(\log(T)^{-1})$. 
\end{prop}

Both propositions exploit a Bernstein-type inequality for $\alpha$-mixing sequences from \citet[Theorem 2.1]{Liebscher:1996} (based on the famous covariance inequality of \citet{Rio:1995}).
Proposition \ref{lem:1} uses a covering argument similar to the one used in \citet{JiangTanner:2010} and \citet{Hansen:2008}. 
Proposition \ref{lem:2} relies on the Bernstein-type inequality and a classic truncation trick used in, for instance, \citet{Hansen:2008}.
See also \citet{Dendramis:Giraitis:Kapetanios:2021} and \citet{Babii:Ghysels:Striaukas:2021} for recent developments on concentration inequalities for dependent data with applications to large-dimensional estimation problems.

\section{Dependent Identically Distributed Observations}\label{sec:homo}

The constant term in the bound of Theorem \ref{thm:erm} can be improved by assuming stationarity.

\begin{asm}[Stationarity]\label{asm:stat}
The sequence of random vectors $\{ (Y_t,\bm X_t')' \}_{t=1}^T$ is stationary.
\end{asm}

We remark that Assumptions \ref{asm:mixing} and \ref{asm:stat} imply that the data are ergodic.

In the stationary case it is convenient to state the moment assumption differently.

\begin{asmredef}{1}{Moments}\label{asm:moments:stat}
The sequences $\{Y_t\}_{t=1}^T$, $\{ \bm X_{t} \}_{t=1}^T$, $\{f^*_t\}_{t=1}^T$ and $\{ \bm Z_{t} \}_{t=1}^T$ with $\bm Z_t = \bm \Sigma^{-{1\over 2}}_t \bm X_t$
satisfy 
$\| Y_t \|_{L_{r_m}} \leq K_m $,
$\sup_{1 \leq i \leq p} \| X_{i\,t} \|_{L_{r_m}} \leq K_m $ and
$\sup_{1 \leq i \leq p} \| (Y_{t}-f^*_t) Z_{i\,t} \|_{L_{r_m}} \leq K_m $, 
for some $K_m \geq 1$ and $r_m > 2$.
\end{asmredef}

The difference between \ref{asm:moments} and \ref{asm:moments:stat} is that the former assumption bounds the $r_m$-th moment of 
$(Y_{t}-f^*_t) X_{i\,t}$ whereas the latter bounds the $r_m$-th moment of $(Y_{t}-f^*_t) Z_{i\,t}$.

In the stationary case the assumption on the eigenvalues of $\bm \Sigma$, Assumption \ref{asm:eigen}, can be dropped. 
In fact, as we show in the proof of Theorem \ref{thm:erm:stationary}, \ref{asm:moments:stat} implies that $\lambda_{\max} \left( \bm \Sigma \right) \leq K_m^2 p $.
This allows the set of predictors to be generated by a factor model \citep{Forni:Hallin:Lippi:Reichlin:2000,Stock:Watson:2002,Bai:Ng:2002,Onatski:2012}.
Additionally, $\lambda_{\min}(\bm \Sigma)$ is allowed to be zero.
This allows the set of predictors to contain some predictors that are perfectly correlated. 

Before stating the main result of this section we introduce a new constant
\[
	K'_{\sigma^2} = K_m^2 \left( 1 + 32 {r_m \over r_m -2} \sum_{l=1}^\infty \alpha(l)^{1-{2\over r_m} } \right) ~.
\]
This constant plays the same role as $K_{\sigma^2}$ and  can be interpreted as an upper bound on the diagonal elements of $\operatorname{Var}\left( {1\over \sqrt{T} } \sum_{t=1}^T (Y_t - f^*_t) \bm Z_t \right)$.
Note that $K'_{\sigma^2} \leq K_{\sigma^2}$. 

We can now state the main result of this section.

\begin{thm}\label{thm:erm:stationary}
Suppose Assumptions \ref{asm:moments:stat},\ref{asm:mixing}--\ref{asm:nparams}, \ref{asm:distribution}--\ref{asm:stat}  are satisfied.
Then, for all $T$ sufficiently large, the empirical risk minimizer defined in \eqref{eqn:erm} satisfies
\[
 	R( \hat {\bm\theta} ) \leq R( \bm \theta^* ) +  K'_{\sigma^2} \left({ 48 \over \kappa_1^2 \kappa_2 }\right)^2 {p \log (T) \over T} ~,
\]
with probability at least $1-{3 K_p K_{m}^{r_m} /( (K'_{\sigma^2})^{1\over 2} \log (T) ) } - o( \log(T)^{-1} )$.
\end{thm}

Note that the bound in Theorem \ref{thm:erm:stationary} does not depend on the eigenvalues of $\bm \Sigma$.
Inspection of the proof shows that if $\{ (Y_t,\bm X_t')' \}$ is an i.i.d.~sequence, $Y_t-f^*_t$ is independent of $\bm X_t$ and $\operatorname{Var}( Y_t - f^*_t ) = \sigma^2$ the bound in Theorem \ref{thm:erm:stationary} becomes
\[
	R( \bm \theta^* ) +  \sigma^2 \left({ 48 \over \kappa_1^2 \kappa_2 }\right)^2 {p \log (T) \over T} ~,
\]
which is close to the bound established in \citet[Corollary 1.2]{LecueMendelson:2016}.





\section{Additional Results}\label{sec:ext}

\paragraph{Alternative Risk Definition.}
It is important to emphasize that the performance measure defined in \eqref{eqn:risk:erm} is the average risk of the prediction rule over the data $\mathcal D$ when $\hat {\bm\theta}$ is estimated using an independent copy of the data $\mathcal D'$.
This measure may have limited appeal for time series applications since a forecaster typically does not have access to an independent copy of the data. 
Alternative more appropriate risk measures may be introduced to evaluate the performance of the risk minimizer in a time series context. 

Assume that we are interested in predicting the out-of-sample observations $\{ (Y_{t}, \bm X_t')' \}_{t = T+1}^{T+H}$ on the basis of the prediction rule estimated from the in-sample observations $\{ (Y_t, \bm X_t')' \}_{t=1}^T$.
For simplicity, here we focus only on the case of identically distributed observations.
We define the average out-of-sample risk of $\bm \theta$ as
\[
		R_\mathsf{oos}(\bm \theta) = \mathbb E\left[ {1 \over H} \sum_{t=T+1}^{T+H} (Y_t - f_{\bm \theta\,t} )^2 \right]  ~,
\]
and we measure the accuracy of the empirical risk minimizer $\hat {\bm \theta}$ using the conditional out-of-sample average risk defined as
\[
		R_\mathsf{oos}(\hat {\bm \theta}) = \mathbb E\left[ \left. {1 \over H} \sum_{t=T+1}^{T+H} (Y_t - f_{\hat {\bm \theta}\,t} )^2 \right| (Y_T,\bm X_T')' , \ldots , (Y_1,\bm X_1')' \right] ~.
\]

For the following result, a slightly stronger version of Assumption \ref{asm:moments} is needed.

\begin{asmredef2}{1}{Moments}\label{asm:moments3}
The sequences $\{Y_t\}_{t=1}^T$, $\{f^*_t\}_{t=1}^T$ and $\{ \bm X_{t} \}_{t=1}^T$
satisfy $\sup_{1 \leq i \leq p} \sup_{1\leq t\leq T} \| Y_{t}-f^*_t \|_{L_{r_m}} \leq K_m $ 
$\sup_{1\leq t\leq T}  \| Y_t \|_{L_{r_m}} \leq K_m $ and $ \sup_{1 \leq i \leq p} \sup_{1\leq t\leq T}  \| X_{i\,t} \|_{L_{r_m}} \leq K_m $ for some $K_m \geq 1$ and $r_m > 4$.
\end{asmredef2}

If we define $K_H = 24 (K_m^2/\underline \lambda) \sum_{l=1}^{\infty} \alpha(l)^{1\over 2}$, we can establish the following theorem.

\begin{thm}\label{thm:erm:oos}
Suppose Assumptions \ref{asm:moments3}, \ref{asm:mixing}, \ref{asm:nparams}, 
\ref{asm:eigen}(i), \ref{asm:distribution}, \ref{asm:smallball} and \ref{asm:stat} are satisfied.
Then, for all $T$ sufficiently large, the empirical risk minimizer defined in \eqref{eqn:erm} satisfies

\[
 	R_\mathsf{oos}( \hat {\bm\theta} ) \leq 
	R_\mathsf{oos}( \bm \theta^* ) +  K'_{\sigma^2} \left({ 48 \over \kappa_1^2 \kappa_2 }\right)^2 {p \log (T) \over T} 
	+ K_H  {p \log (T) \over H} ~,
\]
with probability at least $1-{(6 K_p K_{m}^{r_m} + 1 )/( (K'_{\sigma^2})^{1\over 2} \log (T) ) } - o( \log(T)^{-1} )$.
\end{thm}

A key ingredient in the proof of Theorem \ref{thm:erm:oos} is Ibragimov's inequality \citep{Ibragimov:1962}, 
which bounds the expected value of the difference between the conditional and unconditional
expectation as a function of the $\alpha$-mixing coefficients.
{It is important to remark that the theorem requires $ ( p \log(T) ) / H \rightarrow 0$ in order to have that $ R_\mathsf{oos}( \hat {\bm\theta} ) - R_\mathsf{oos}( \bm \theta^* ) \rightarrow 0$.
In other words there exists a ``wedge'' between $R_\mathsf{oos}( \hat {\bm\theta} )$ and $R_\mathsf{oos}( {\bm\theta}^* )$ that only vanishes as the forecast horizon $H$ grows large.
This may be intuitively explained as follows.
The empirical risk minimizer $\hat{\bm \theta}$ is consistent for $\bm \theta^*$, the minimizer of $R(\bm \theta)$.
However, the minimizers of $R(\bm \theta)$ and $R_{oos}(\bm \theta)$ are not guaranteed to be same for finite $H$ and the difference between the two only vanishes as the forecast horizon $H$ grows large.}

\paragraph{Small-Ball Assumption.} It is possible to introduce alternative assumptions that imply the small-ball condition stated in \ref{asm:smallball}.
For example, as \citet{LecueMendelson:2016} remark, the small-ball condition holds when the $L_2$ and $L_4$ norms of $f_{\bm \theta\,t}$ are equivalent. 
More precisely, if for each $\bm \theta \in \mathbb R^p$ (and all $t=1,\ldots,T$) it holds that \(\| f_{\bm \theta\,t} \|_{L_4} \leq C \| f_{\bm \theta\,t} \|_{L_2} \),
for some constant $C$ (that does not depend on $\bm \theta$ or $t$). We remark that norm equivalence conditions are commonly used in the literature to establish the properties of empirical risk minimization, see for instance \citet{AudibertCantoni:2011}.

To give a concrete example, below we show that if the distribution of the standardized predictors $\bm Z_t = \bm \Sigma^{-{1\over 2}} \bm X_t$ is spherical then the small-ball assumption is satisfied.

\begin{asmredef}{6}{Spherical Density}\label{asm:spherical}
Consider the sequence of random vectors $\{ \bm Z_t \}_{t=1}^T$ with $\bm Z_t = \bm \Sigma_t^{-{1\over 2}} \bm X_t$.
Then for each $t=1,\ldots,T$ it holds that $\bm Z_t \sim \bm S$ where $\bm S$ is a $p$-dimensional spherical random vector that satisfies $\sup_{1\leq i \leq p} \| S_i \|_{L_4} < \infty$.
\end{asmredef}

\begin{lemma}\label{lemma:smallball} 
Suppose Assumption \ref{asm:spherical} holds.
Then, for each $t=1,\ldots,T$ and for each $\bm \theta_1, \bm \theta_2 \in \mathbb R^p$, 
\(
	\mathbb P \left( | f_{\bm \theta_1\,t} - f_{\bm \theta_2\,t} | \geq \kappa_1 \| f_{\bm \theta_1\,t} - f_{\bm \theta_2\,t} \|_{L_2} \right) \geq \kappa_2 
\)
holds for some $\kappa_1>0$ and $\kappa_2>0$.
\end{lemma}

The proof uses the Paley-Zygmund inequality.
\ref{asm:spherical} can replace \ref{asm:smallball} in Theorem \ref{thm:erm} and Theorem \ref{thm:erm:stationary}.

\section{Conclusion}\label{sec:end}

This paper establishes oracle inequalities for the prediction risk of the empirical risk minimizer for large-dimensional linear regression.
We generalize existing results by allowing the data to be dependent and heavy-tailed.
Our main results show that the empirical risk minimizer achieves optimal performance (up to a logarithmic factor). 
The results have been established using the small-ball method, which is a powerful technique to obtain oracle inequalities. 
Future research includes extending these results to regularized empirical risk minimization, analogously to \citet{LecueMendelson:2017,LecueMendelson:2018}.

\appendix

\section{Proofs}

\begin{proof}[Proof of Lemma \ref{lemma:opt}]
$(i)$ The existence of $\bm \theta^*$ follows from the fact that $R(\bm \theta)$ is quadratic. \\
$(ii)$ It is equivalent to show that $\bm \theta^*$ satisfies
\begin{equation}\label{eqn:opt:orthogonal}
	{1 \over T}\sum_{t=1}^T  \mathbb E [ (Y_t-f^*_{t} )( f^*_{t} - f_{\bm \theta\,t} ) ] = \left({1 \over T}\sum_{t=1}^T  \mathbb E [ (Y_t-f^*_{t} ) \bm X_t' ]\right)( \bm \theta^* - \bm \theta) = 0 ~,
\end{equation}
for any $\bm \theta \in \mathbb R^p$. We then have that \eqref{eqn:opt:orthogonal} 
is implied by the first order condition for a minimum for $R(\bm \theta)$, 
as $\bm \theta^*$ is such that  
\( {2 \over T}\sum_{t=1}^T  \mathbb E [ (Y_t-f_t^* ) \bm X'_t ] = 0 \).\\
$(iii)$ This follows from the strict convexity of $R(\bm \theta)$. 
\end{proof}

\begin{proof}[Proof of Theorem \ref{thm:erm}]
Define the empirical risk differential for an arbitrary $\bm \theta \in \mathbb R^p$ as
\begin{equation*}
	\widehat{\mathcal L}_{\bm \theta} = R_T(\bm \theta) - R_T(\bm \theta^*) = {1 \over T}\sum_{t=1}^T ( f^*_{t} - f_{\bm \theta\,t} )^2 + {2 \over T} \sum_{t=1}^T (Y_t-f^*_{t} )( f^*_{t} - f_{\bm \theta\,t} ) ~.
\end{equation*}	
Assume that it holds that 
\begin{equation}\label{eqn:mu_is_far}
	 {1 \over T} \sum_{t=1}^T \|f^*_t - f_{\bm \theta\,t} \|_{L_2}  > 48 \sqrt{ K_{\sigma^2} \over \underline{\lambda} } {K_{\bm \Sigma} \over  \kappa_1^2 \kappa_2 } \sqrt{ {p \log (T) \over T}} ~.
\end{equation}
Conditioning on the events of Proposition \ref{lem:1} and Proposition \ref{lem:2},
for all $T$ sufficiently large, at least with probability $1 - {3 K_p (2K_m)^{r_m} /(K^{1\over 2}_{\sigma^2} \log (T)) } - o(\log(T)^{-1})$, we have that
\begin{align*}
	{1 \over T}\sum_{t=1}^T ( f^*_{t} - f_{\bm \theta\,t} )^2 
	& \stackrel{(a)}{\geq} { \kappa_1^2 \kappa_2 \over 2 K_{\bm \Sigma}} {1 \over  T} \sum_{t=1}^T \| f^*_t - f_{\bm \theta\,t} \|^2_{L_2} 
	\stackrel{(b)}{\geq} { \kappa_1^2 \kappa_2 \over 2 K_{\bm \Sigma}} {1 \over T} \sum_{t=1}^T \|f^*_t - f_{\bm \theta\,t} \|_{L_2} {1 \over T} \sum_{t=1}^T \| f^*_t - f_{\bm \theta\,t}  \|_{L_2}\\
	& \stackrel{(c)}{>} 24 \sqrt{ K_{\sigma^2} \over \underline \lambda } {1\over T}  \sum_{t=1}^T \|f^*_t - f_{\bm \theta\,t} \|_{L_2} \sqrt{ {p \log (T) \over T } } 
 	\stackrel{(d)}{\geq} \left| {2 \over T}\sum_{t=1}^T (Y_t-f^*_{t} )( f^*_{t} - f_{\bm \theta\,t} ) \right| ~,
\end{align*}	
where 
$(a)$ follows from Proposition \ref{lem:1},
$(b)$ follows from Jensen's inequality,
$(c)$ follows from condition \eqref{eqn:mu_is_far},
and $(d)$ follows from Proposition \ref{lem:2}.
Thus, conditional on the events of Proposition \ref{lem:1} and Proposition \ref{lem:2} and assuming \eqref{eqn:mu_is_far} holds we have with high probability that \( \widehat{\mathcal L}_{\bm \theta} > 0 \).
Since the empirical risk minimizer $\hat{\bm \theta}$ satisfies $\widehat{\mathcal L}_{\hat{\bm \theta}} \leq 0$ then conditional on the same events we have $ {1 \over T} \sum_{t=1}^T \|f^*_t - \hat f_{t} \|_{L_2} \leq { 48 K^{1\over 2}_{\sigma^2} K_{\bm \Sigma} \over \underline{\lambda}^{1\over 2} \kappa_1^2 \kappa_2 } \sqrt{ {p \log (T)\over T} } $. 
The claim follows from
\begin{equation}\label{eqn:final}
	R(\hat{\bm\theta}) - R(\bm \theta^*) 
	= {1 \over T} \sum_{t=1}^T \| f^*_t - \hat f_{t} \|^2_{L_2} 
	\leq \overline{\lambda} \| \hat{\bm \theta}-\bm \theta^* \|^2_{2} 
	\leq { K_{\sigma^2} } \, { K^{3}_{\bm \Sigma} \over \underline{\lambda} }  \,\left({ 48 \over
    \kappa_1^2 \kappa_2 }\right)^2 \, { {p \log (T) \over T }  } ~,
\end{equation}
where
the first equality follows from Lemma \ref{lem:1} where the $L_2$ norm is conditional on $\{ \hat \theta = \hat \theta( \mathcal D' )\}$,
the first inequality follows from $ {1 \over T} \sum_{t=1}^T \|f^*_t - \hat f_{t} \|^2_{L_2} \leq \overline{\lambda} \| \hat{\bm \theta} - \bm \theta^* \|^2_{2} $, and
the second inequality follows from $ \underline{\lambda}^{1\over 2} \| \hat{\bm \theta}-\bm \theta^* \|_{2} \leq{1 \over T} \sum_{t=1}^T \|f^*_t - \hat f_{t} \|_{L_2} $.
\end{proof}

\begin{proof}[Proof of Proposition \ref{lem:1}]
For any $\bm \theta \in \mathbb R^p \setminus \{ \bm \theta^* \} $ (notice that \ref{asm:eigen} implies that $\bm \theta^*$ is unique),
define the standardized parameter vector \( \bm v = { (\bm \theta_*-\bm \theta) / \sqrt{ {1 \over T }\sum_{t=1}^T \| f^*_{t} - f_{\bm \theta\,t} \|^2_{L_2}  } } \) and note 
\begin{align*}
	& {1 \over T} \sum_{t=1}^T ( f^*_{t} - f_{\bm \theta\,t} )^2  =  
	{ {1 \over T} \sum_{t=1}^T ( f^*_{t} - f_{\bm \theta\,t} )^2 \over {1 \over T} \sum_{t=1}^T\| f^*_{t} - f_{\bm \theta\,t} \|^2_{L_2}  } {1 \over T} \sum_{t=1}^T \| f^*_{t} - f_{\bm \theta\,t} \|^2_{L_2} \\
	& \quad = 
	{1 \over T} \sum_{t=1}^T ( \bm X_t' \bm v )^2 \; {1 \over T} \sum_{t=1}^T \| f^*_{t} - f_{\bm \theta\,t} \|^2_{L_2} 
	\geq  
	{\kappa_1^2 \over K_{\bm \Sigma} } {1 \over T } \sum_{t=1}^T \mathbbm{1}_{ \left\{ | \bm X_t' \bm v| \geq \kappa_1 K_{\bm \Sigma}^{-1/2}  \right\} } \; {1 \over T} \sum_{t=1}^T \| f^*_{t} - f_{\bm \theta\,t} \|^2_{L_2} ~.
\end{align*}	
Let \( g_{\bm v\,t} = \mathbbm{1}_{ \{ |\bm X_t' \bm v| \geq \kappa_1 K_{\bm \Sigma}^{-1/2} \} } \), define $V = \{ \bm v \in \mathbb R^p : {1\over T} \mathbb E [ \sum_{t=1}^T ( \bm X_t' \bm v )^2 ] = 1 \} $ and note that
\[
	{1 \over T} \sum_{t=1}^T g_{\bm v\,t}
	={1 \over T} \sum_{t=1}^T \mathbb E g_{\bm v\,t} + g_{\bm v\,t} - \mathbb E g_{\bm v\,t} 
	\geq{1 \over T} \sum_{t=1}^T \mathbb Eg_{\bm v\,t} - \sup_{\bm v \in V} \left| {1 \over T} \sum_{t=1}^T  g_{\bm v\,t} - \mathbb E g_{\bm v\,t} \right| ~,
\]
since the standardized parameter vector $\bm v$ belongs to $V$.
Let $V_i = \{ \bm v \in \mathbb R^p : \|\bm v-\bm v_i\|_2 \leq \delta \}$ with $\bm v_i \in V$ for $i=1,\ldots,N_\delta$ denote a $\delta$-covering of $V$. 
Then, we have that
\begin{align*}
	& 
	\mathbb P\left( \sup_{\bm v \in V} \left| {1\over T} \sum_{t=1}^T  g_{\bm v\,t}- \mathbb E g_{\bm v\,t} \right| > \varepsilon \right) 
	\leq \sum_{i=1}^{N_\delta} \mathbb P\left( \sup_{\bm v \in V_i} \left| {1\over T} \sum_{t=1}^T g_{\bm v\,t} - \mathbb E g_{\bm v\,t} \right| > \varepsilon \right) \\
	&
	\quad \leq \sum_{i=1}^{N_\delta} \mathbb P\left( \left| {1\over T} \sum_{t=1}^T g_{i\,t}- \mathbb E g_{i\,t} \right| > {\varepsilon\over 2} \right) \\
	&
	\quad + \sum_{i=1}^{N_\delta} \mathbb P\left( \sup_{\bm v \in V_i} \left| {1\over T} \sum_{t=1}^T g_{\bm v\,t}- \mathbb Eg_{\bm v\,t} 
	- \left( {1\over T} \sum_{t=1}^T g_{i\,t} - \mathbb E g_{i\,t} \right) \right| > {\varepsilon \over 2} \right) ~,
\end{align*}
where $g_{i\,t} = g_{\bm v_i\,t}$.
Proposition \ref{prop:l1:S} establishes that $(i)$ for each $\bm v \in V_i$ we have \( |g_{\bm v\,t} -  g_{i\,t}| \leq \bar g_{i\,t} \) where $ \bar g_{i\,t}$ is defined in that proposition
and $(ii)$ there exists a positive constant $K_1$ (that does not depend on $i$, $t$ and $p$) such that for all $\delta < {K_{\bm \Sigma}^{-1/2} / ( 2 \overline{\lambda}^{1/2} ) }$ we have that $\mathbb E \bar g_{i\,t} \leq K_1 p^{1/2} \delta$.
Set $\delta = \varepsilon / (8 K_1 p^{1/2})$ and note that for all $\varepsilon < 4 {K_{\bm \Sigma}^{-1/2} K_1 p^{1/2}/ \overline{\lambda}^{1/2} }$ 
\begin{align*}
	& \mathbb P\left( \sup_{\bm v \in V_i} \left| {1\over T} \sum_{t=1}^T g_{\bm v\,t}- \mathbb Eg_{\bm v\,t} 
	- \left( {1\over T} \sum_{t=1}^T g_{i\,t} - \mathbb E g_{i\,t} \right) \right| > {\varepsilon \over 2} \right) \\
	& \quad =
	\mathbb P\left( \sup_{\bm v \in V_i} \left| {1\over T} \sum_{t=1}^T ( g_{\bm v\,t} - g_{i\,t} )
	- ( \mathbb E g_{\bm v\,t} - \mathbb E g_{i\,t} ) \right| > {\varepsilon \over 2} \right) \leq \mathbb P\left( {1\over T} \sum_{t=1}^T (\bar g_{i\,t} + \mathbb E \bar g_{i\,t}) > { \varepsilon \over 2} \right) \\
	& \quad = \mathbb P\left( {1\over T} \sum_{t=1}^T (\bar g_{i\,t} - \mathbb E \bar g_{i\,t}) > { \varepsilon \over 2} - {2\over T} \sum_{t=1}^T \mathbb E \bar g_{i\,t} \right) 
	\stackrel{(a)}{\leq} \mathbb P\left( {1\over T} \sum_{t=1}^T (\bar g_{i\,t}- \mathbb E \bar g_{i\,t}) > {\varepsilon \over 4} \right)  ~,
\end{align*}
where $(a)$ follows from the fact that $\mathbb E \bar g_{i\,t} \leq \varepsilon/8$.
Finally, we have that 
\begin{align*}
	& \mathbb P\left( \sup_{\bm v \in V} \left| {1\over T} \sum_{t=1}^T g_{\bm v\,t} - \mathbb E g_{\bm v\,t} \right| > \varepsilon \right) \\
	& \quad \leq \sum_{i=1}^{N_\delta} \mathbb P\left( \left| {1\over T} \sum_{t=1}^T g_{i\,t}- \mathbb E g_{i\,t} \right| > {\varepsilon\over 2} \right) 
	+ \sum_{i=1}^{N_\delta} \mathbb P\left( \left| {1\over T} \sum_{t=1}^T \bar g_{i\,t} - \mathbb E \bar g_{i\,t} \right| > {\varepsilon \over 4} \right) \\
	& \quad \leq N_\delta \max_{1 \leq i \leq N_{\delta}} \mathbb P\left( \left| {1\over T} \sum_{t=1}^T Z'_{i\,t} \right| > {\varepsilon\over 2} \right) 
	+ N_\delta \max_{1 \leq i \leq N_{\delta}}\mathbb P\left( \left| {1\over T} \sum_{t=1}^T Z''_{i\,t} \right| > {\varepsilon \over 4} \right) ~,
\end{align*}
where $Z_{i\,t}' = g_{i\,t}- \mathbb E g_{i\,t} $ and $Z_{i\,t}'' = \bar g_{i\,t} - \mathbb E \bar g_{i\,t}$.
We have that
\begin{align*}
    N_\delta \max_{1 \leq i \leq N_{\delta}} 
    \mathbb P\left( \left| {1\over T} \sum_{t=1}^T Z'_{i\,t} \right| > {\varepsilon\over 2} \right) 
    \stackrel{(a)}{\leq} \left( 1+{16 K_1 p^{1/2} \over \underline{\lambda}^{1/2} \varepsilon } \right)^p 
    \max_{1 \leq i \leq N_{\delta}} 
    \mathbb P\left( \left| {1\over T} \sum_{t=1}^T Z'_{i\,t} \right| > {\varepsilon\over 2} \right) ~,
\end{align*}
where $(a)$ follows from the fact that the $\delta$-covering number $N_\delta$ of an Euclidian sphere of radius $C$ in $\mathbb R^p$ satisfies $N_\delta \leq (1 + (2 C)/ \delta )^p$ \citep[Corollary 4.2.13]{vershynin:2018},
and that the covering number of $V$ is smaller than the covering number of $\{ \bm v \in \mathbb R^p: \| \bm v \|_2 \leq \underline{\lambda}^{-1/2} \}$ since $V \subset \{ \bm v \in \mathbb R^p: \| \bm v \|_2 \leq \underline{\lambda}^{-1/2} \}$.
Note that $Z'_{i\,t}$ inherits the mixing properties of $(Y_t,\bm X_t')'$ and satisfies $\|Z'_{i\,t}\|_{L_\infty} \leq 1$.
It follows from Proposition \ref{prop:l1:conc} that for all $T$ sufficiently large and for the choice of \( \varepsilon'_T \) spelled out in that proposition
that $\varepsilon'_T \leq 4 {K_{\bm \Sigma}^{-1/2} K_1 p^{1/2} / \overline{\lambda}^{1/2} } \wedge \kappa_2/2 $ and 
\begin{equation}\label{eqn:l1:conc1}
	\left( 1+{16 K_1 p^{1/2} \over \underline{\lambda}^{1/2} \varepsilon'_T } \right)^p 
    \max_{1 \leq i \leq N_{\delta}} 
    \mathbb P\left( \left| {1\over T} \sum_{t=1}^T Z'_{i\,t} \right| > { \varepsilon_T' \over 2} \right) 
	\leq {4 \over T} + o\left({1\over T}\right) ~.
\end{equation}
Using analogous arguments, we have that for all $T$ sufficiently large and for the choice of \( \varepsilon''_T \) spelled out in Proposition \ref{prop:l1:conc}
that $\varepsilon''_T \leq 4 {K_{\bm \Sigma}^{-1/2} K_1 p^{1/2} / \overline{\lambda}^{1/2} } \wedge \kappa_2/2 $ and 
\begin{equation}\label{eqn:l1:conc2}
	N_\delta \max_{1 \leq i \leq N_{\delta}} 
	 \mathbb P\left( \left| {1\over T} \sum_{t=1}^T Z''_{i\,t} \right| > {\varepsilon''_T \over 4} \right) \leq
	\left( 1+{16 K_1 p^{1/2} \over \underline{\lambda}^{1/2} \varepsilon''_T } \right)^p 
    \max_{1 \leq i \leq N_{\delta}} 
    \mathbb P\left( \left| {1\over T} \sum_{t=1}^T Z''_{i\,t} \right| > { \varepsilon_T'' \over 4} \right) 
	\leq {4 \over T} + o\left({1 \over T}\right) ~.
\end{equation}
The inequalities in \eqref{eqn:l1:conc1} and \eqref{eqn:l1:conc2} imply that for all $T$ sufficiently large we can pick $\varepsilon_T = \varepsilon_T' \wedge \varepsilon_T''$ to obtain
\begin{equation*}
	\sup_{\bm v \in V} \left| {1 \over T} \sum_{t=1}^T  g_{\bm v\,t} - \mathbb E g_{\bm v\,t} \right|
	\leq { \kappa_2 \over 2 } ~.
\end{equation*}
with probability at least $1-8T^{-1} - o(T^{-1})$.
The claim of the proposition follows after noting that with probability at least $1-8T^{-1} - o(T^{-1})$ we have
\begin{align*}
	& { \kappa_1^2 \over K_{\bm \Sigma} } {1 \over T} \sum_{t=1}^T \mathbbm{1}_{ \{ | \bm X_t' \bm v| \geq \kappa_1 K_{\bm \Sigma}^{-1/2}\} } \; {1 \over T} \sum_{t=1}^T \| f^*_{t} - f_{\bm \theta\,t} \|^2_{L_2} \\
	& \quad \geq { \kappa_1^2 \over K_{\bm \Sigma} } \left( {1 \over T} \sum_{t=1}^T \mathbb P( | \bm X_t' \bm v| \geq \kappa_1 K_{\bm \Sigma}^{-1/2} ) - \sup_{\bm v \in V} \left| {1 \over T} \sum_{t=1}^T  g_{\bm v\,t} - \mathbb E g_{\bm v\,t} \right| \right)
	\; {1 \over T} \sum_{t=1}^T \| f^*_{t} - f_{\bm \theta\,t} \|^2_{L_2} \\
	& \quad \stackrel{(a)}{\geq} {\kappa_1^2 \over K_{\bm \Sigma} } \left( \kappa_2 - { \kappa_2 \over 2 } \right) {1 \over T} \sum_{t=1}^T \| f^*_{t} - f_{\bm \theta\,t} \|^2_{L_2} 
	\geq {\kappa_1^2 \kappa_2 \over 2 K_{\bm \Sigma} } {1 \over T} \sum_{t=1}^T \| f^*_t - f_{\bm \theta\,t} \|^2_{L_2} ~,
\end{align*}
where $(a)$ follows from the fact that
$ \mathbb P( | \bm X_t' \bm v| \geq \kappa_1 K_{\bm \Sigma}^{-1/2} )  \geq \mathbb P( | \bm X_t' \bm v| \geq \kappa_1 \| \bm X_t'\bm v \|_{L_2} ) \geq \kappa_2 $
since $ \| \bm X_t'\bm v \|_{L_2} = { \| \bm X_t' (\bm \theta^* -\bm \theta )\|_{L_2} / \sqrt{ {1 \over T }\sum_{t=1}^T \| \bm X_t' (\bm \theta^*-\bm \theta) \|^2_{L_2}  } } \geq K^{-1/2}_{\bm \Sigma} $.
\end{proof}

\begin{proof}[Proof of Proposition \ref{lem:2}]
Define \( \bm \nu_t = \mathbb E \left[ (Y_t - f^*_t)\bm X_t  \right] \) and note that Lemma \ref{lemma:opt} implies 
\[
	\sum_{t=1}^T \bm \nu_t' ( \bm \theta^*-\bm \theta) = \mathbb E \left( \sum^T_{t=1}(Y_t-f^*_t)(f^*_t-f_{\bm \theta\,t}) \right) = 0 \text{ for any } \bm \theta \in \mathbb R^p ~.
\]
For any $\bm \theta \in \mathbb R^p \setminus \{ \bm \theta^* \}$ we have that (notice that \ref{asm:eigen} implies that $\bm \theta^*$ is unique)
\[
	\mathbb P \left( { \left| \sum^T_{t=1}(Y_t-f^*_t)(f^*_t-f_{\bm \theta\,t})\right| \over \sum^T_{t=1}\|f^*_t-f_{\bm \theta\,t}\|_{L_2} } > \varepsilon \right) \leq 
	\mathbb P \left( \sup_{\bm \theta \in \mathbb R^p \setminus \{ \bm \theta^* \} } { \left| \sum^T_{t=1}(Y_t-f^*_t)(f^*_t-f_{\bm \theta\,t} ) \right|\over \sum^T_{t=1}\|f^*_t-f_{\bm \theta\,t}\|_{L_2} } > \varepsilon \right) ~.
\]
Define \(\bm v = {(\bm \theta^*-\bm \theta)}/({\frac{1}{T}\sum^T_{t=1} \| f^*_t -f_{\bm \theta\,t}\|_{L_2}} ) \) for any $\bm \theta \in \mathbb R^p \setminus \{ \bm \theta^* \}$ and note that 
\begin{equation*}
	\|\bm v\|_2 = \frac{\|\bm \theta^*-\bm \theta\|_2}{\frac{1}{T}\sum^T_{t=1}\sqrt{(\bm \theta^*-\bm \theta)'\mathbb{E}(\bm X_t \bm X_t')(\bm \theta^*-\bm \theta)}} \leq \underline{ \lambda }^{-1/2}  ~.
\end{equation*}
Then we have that 
\begin{align*}
	\sum^T_{t=1}\frac{\left| (Y_t-f^*_t)(f^*_t-f_{\bm \theta\,t}) \right| }{\sum^T_{t=1}\| f^*_t - f_{\bm \theta\,t} \|_{L_2}} 
	&= \sum^T_{t=1} { \left| (Y_t-f^*_t) \bm X_{t}'(\bm \theta^*-\bm \theta) \right| \over \sum^T_{t=1} \|f^*_t - f_{\bm \theta\,t} \|_{L_2} } 
	= \left| {1 \over T} \sum^T_{t=1} (Y_t-f^*_t) \bm X_{t}' \bm v \right| \\
	&= \left| {1 \over T} \sum^T_{t=1} [ (Y_t-f^*_t) \bm X_{t}' - \bm \nu_t' ] \bm v \right| = \left| { 1 \over T} \sum^T_{t=1} \bm U_{t}'\bm v \right| ~,
\end{align*}
where \( \bm U_t=  {(Y_t-f^*_t)} \bm X_{t} - \bm \nu_{t} \).
Next, we have 
\begin{align*}
	& \mathbb P \left( \sup_{\bm \theta \in \mathbb R^p \setminus \{ \bm \theta^* \} } { \left| \sum^T_{t=1}(Y_t-f^*_t)(f^*_t - f_{\bm \theta\,t} ) \right| \over \sum^T_{t=1}\|f^*_t-f_{\bm \theta\,t}\|_{L_2} } > \varepsilon \right) 
	\leq \mathbb{P}\left( \sup_{\bm v: \|\bm v\|_2 \leq \underline{\lambda}^{-1/2}} \left| \frac{ 1 }{T}\sum^T_{t=1} \bm U'_{t}\bm v \right| >\varepsilon\right) \\
	& \quad \leq  \mathbb{P}\left( \sup_{\bm v: \|\bm v\|_2 \leq \underline{\lambda}^{-1/2}} \left\| \frac{1}{T}\sum^T_{t=1} \bm U_{t} \right\|_2 \|\bm v\|_2 >\varepsilon\right) 
	\leq \mathbb{P}\left(  \left\| \frac{1}{T}\sum^T_{t=1} \bm U_t \right\|_2 >  \underline{\lambda}^{1/2} \varepsilon \right) ~.
\end{align*}
Note that $\{ \bm U_{t} \}$ is mean zero, 
satisfies $\| U_{i\,t} \|_{L_2} \leq K_m$ and
\begin{align*}
	\| U_{i\,t} \|_{L_{r_m}} & \leq \| (Y_t-f^*_t) X_{i\,t} \|_{L_{r_m}} + \| \nu_{i\,t} \|_{L_{r_m}} = \| (Y_t-f^*_t) X_{i\,t} \|_{L_{r_m}} + \| (Y_t-f^*_t) X_{i\,t} \|_{L_{1}} \leq 2 K_m 
\end{align*}
because of \ref{asm:moments},
and inherits the mixing properties of $\{ (Y_t,\bm X_t)' \}$ spelled out in \ref{asm:mixing}.
Proposition \ref{prop:l2:conc} then implies that, for all $T$ sufficiently large, we have 
\[
	\sup_{\bm \theta \in \mathbb R^p \setminus \{ \bm \theta^* \} } { \left| \sum^T_{t=1}(Y_t-f^*_t)(f_{\bm \theta\,t}-f^*_t ) \right| \over \sum^T_{t=1}\|f_{\bm \theta\,t} - f^*_t\|_{L_2} } 
	\leq  12 \, \sqrt{ K_{\sigma^2} \over \underline{\lambda} } \, \sqrt{ {p \log (T) \over T}  } ~,
\]
with probability at least $1- 3 K_p (2 K_m)^{r_m} / ( K^{1\over 2}_{\sigma^2} \log (T) ) - o( \log(T)^{-1} )$ where $K_{\sigma^2}$ is the constant $\sigma^2$ defined in that proposition. The claim of the proposition then follows.
\end{proof}

\begin{proof}[Proof of Theorem \ref{thm:erm:stationary}]
We begin by showing that when \ref{asm:moments} is satisfied we have $\lambda_{\max}(\bm \Sigma) \leq K_m^2 p $.
Let $\bm \Sigma = \mathbb E( \bm X_t \bm X_t' )$ and let $\bm \Sigma_{i \bullet}$ denote the $i$-th row of $\bm \Sigma$.
{Then, 
\begin{equation*}
\lambda_{\max}(\bm \Sigma) = \sup_{\bm{x} \in \mathbb R^p: \lVert\bm{x}\rVert_2 = 1} \lVert\bm \Sigma \bm{x}\rVert_2 
= \sup_{\bm{x} \in \mathbb R^p: \lVert\bm{x}\rVert_2 = 1} \sqrt{\sum_{i = 1}^p ( \bm \Sigma_{i\,\bullet} \bm x)^2  } 
\leq \sup_{\bm{x} \in \mathbb R^p: \lVert\bm{x}\rVert_2 = 1} \sqrt{\sum_{i = 1}^p
\lVert\bm{\Sigma}_{i\, \bullet}\rVert_2^2  \lVert\bm{x}\rVert_2^2  }
\end{equation*}
\begin{equation*}
= \sqrt{\sum_{i = 1}^p \lVert\bm{\Sigma}_{i\, \bullet}\rVert_2^2 } 
\leq \sqrt{\sum_{i = 1}^p \lVert K_m^2 \bm{1}_p\rVert_2^2 } 
= K_m^2 \sqrt{\sum_{i = 1}^p p } = K_m^2 p ~,
\end{equation*}
where $\bm{1}_p$ is $p$-dimensional vector with entries equal to one.} \\
The proof is similar to the one of Theorem \ref{thm:erm} and we only highlight the main differences.\\
In the proof of Proposition \ref{lem:1} define $\bm Z_t = \bm \Sigma^{-{1\over 2}} \bm X_t$ and $\bm v = \bm \Sigma^{1\over 2} (\bm \theta_* - \bm \theta) / \| f_t^* - f_{\bm \theta\,t} \|_{L_2}$.
Then \ref{asm:smallball} and \ref{asm:stat} imply $\mathbb P( | \bm Z_t' \bm v| \geq \kappa_1 \| \bm Z_t'\bm v \|_{L_2} ) \geq \kappa_2 $.
Thus, in that proposition the function $g_{\bm v\,t}$ can be defined as $ \mathbbm 1_{\{ | \bm Z_t' \bm v| \geq \kappa_1 \| \bm Z_t'\bm v \|_{L_2} \}} $.
Proposition \ref{prop:l1:S} can then be modified and it is straightforward to see that there exists a $\bar g_{i\,t}$ function such that for all $\delta < 1/2 $ we have
$ \mathbb E \bar g_{i\,t} \leq K_1 p^{1\over 2} \delta$ for some positive constant $K_1$. If we set $\delta = \varepsilon / (8 K_1 p^{1\over 2}) $ for all $\varepsilon < 4 K_1 p^{1\over2} $ we get, following the same steps as in Proposition \ref{lem:1} and noting that $\| \bm v \|_2=1$, that
\begin{align*}
	& \mathbb P\left( \sup_{\bm v \in V} \left| {1\over T} \sum_{t=1}^T g_{\bm v\,t} - \mathbb E g_{\bm v\,t} \right| > \varepsilon \right) \\
	& \quad \leq 
	\left( 1+{16 K_1 p \over \varepsilon } \right)^p 
	\max_{i=1,\ldots,N_\delta} \left[ \mathbb P\left( \left| {1\over T} \sum_{t=1}^T Z'_{i\,t} \right| > { \varepsilon \over 2} \right) + \mathbb P\left( \left| {1\over T} \sum_{t=1}^T Z''_{i\,t} \right| > { \varepsilon \over 4} \right) \right] ~.
\end{align*}
Finally, Proposition \ref{prop:l1:conc} implies that for all $T$ sufficiently large and any $\bm \theta \in \mathbb R^p$,
\[
	{1 \over T}\sum_{t=1}^T ( f^*_{t} - f_{\bm \theta\,t} )^2 \geq { \kappa_1^2 \kappa_2 \over 2} {1 \over T} \sum_{t=1}^T \| f^*_t - f_{\bm \theta\,t} \|^2_{L_2} ~
\]
holds with probability at least $1-8 T^{-1} - o(T^{-1})$.\\
{In the proof of Proposition \ref{lem:2} define $\bm v = \bm \Sigma^{1\over 2} {(\bm \theta^*-\bm \theta)}/\| f^*_t -f_{\bm \theta\,t}\|_{L_2}$ and $\bm U_{t} = (Y_t - f^*_t) \bm Z_t$.
Following the steps of the proof of Proposition \ref{lem:2} we have that for any $\bm \theta \in \mathbb R^p \setminus \{ \bm \theta^* \}$
\begin{align*}
	& \mathbb P \left( \sup_{\bm \theta \in \mathbb R^p \setminus \{ \bm \theta^* \} } { \left| \sum^T_{t=1}(Y_t-f^*_t)(f^*_t - f_{\bm \theta\,t} ) \right| \over \sum^T_{t=1}\|f^*_t-f_{\bm \theta\,t}\|_{L_2} } > \varepsilon \right) 
	\leq \mathbb{P}\left(  \left\| \frac{1}{T}\sum^T_{t=1} \bm U_t \right\|_2 > \varepsilon \right) ~,
\end{align*}
where we have used the fact that $\| \bm v \|_2 = 1$.
Note that $\{ \bm U_t \}$ is mean zero,
satisfies $\| U_{i\,t} \|_{L_{r_m}} \leq K_m$ for each $i=1,\ldots,p$ because of \ref{asm:moments:stat}
and inherits the mixing properties of $\{ (Y_t,\bm X_t)' \}$ spelled out in \ref{asm:mixing}.}
Applying Proposition \ref{prop:l2:conc} we have that for all $T$ sufficiently large and any $\bm \theta \in \mathbb R^p \setminus \{ \bm \theta^* \}$,
\[
	{1 \over T} \sum^T_{t=1}(Y_t-f^*_t)(f^*_t- f_{\bm \theta\,t} ) \leq  12 \, \sqrt{ K'_{\sigma^2} } {1 \over T}\sum^T_{t=1}\| f^*_t- f_{\bm \theta\,t}\|_{L_2} \, \sqrt{ {p \log (T) \over T} } 
\]
holds with probability at least $1- 3 K_p K_{m}^{r_m} / ( (K_{\sigma^2}')^{1\over 2} \log (T) ) - o(\log(T)^{-1})$ with \( K'_{\sigma^2} = K_m^2 \left( 1 + 32 {r_m \over r_m -2} \sum_{l=1}^\infty \alpha(l)^{1-{2\over r_m} } \right) \).\\
Finally, in the proof of Theorem \ref{thm:erm} we can replace condition \eqref{eqn:mu_is_far} with 
\begin{equation*}
	 \| f^*_t - f_{\bm \theta\,t} \|_{L_2}  > { 48 (K'_{\sigma^2})^{1\over 2} \over \kappa_1^2 \kappa_2 } \sqrt{ {p \log (T) \over T}} ~.
\end{equation*}
Following the same steps as in the proof there we obtain the claim.
\end{proof}

\begin{proof}[Proof of Theorem \ref{thm:erm:oos}]
We begin by introducing the out-of-sample risk for the ``ghost'' out-of-sample observations.
Let  $\{ (Y^G_t, (\bm X^G_t)')' \}_{t=T+1}^{T+H}$ denote a sequence of observations from the $\{ (Y_t, \bm X_t')' \}$ process that is independent of $\{ (Y_t, \bm X_t')' \}_{t=1}^{T}$.
Then define 
\begin{align*}
	R^G_\mathsf{oos}( {\bm \theta}) & =
	\mathbb E\left[ {1 \over H} \sum_{t=T+1}^{T+H} (Y^G_t - f^G_{\bm \theta\,t} )^2 \right] \\
	R^G_\mathsf{oos}( \hat{\bm \theta}) & =
	\mathbb E\left[ \left. {1 \over H} \sum_{t=T+1}^{T+H} (Y^G_t - \hat f^G_t )^2 \right| (Y_T,\bm X_T')' , \ldots , (Y_1,\bm X_1')' \right] ~,
\end{align*}
where $f^G_{\bm \theta\,t} = {\bm \theta}' \bm X^G_t$ and $\hat f^G_{t} = \hat{\bm \theta}' \bm X^G_t$ with $\hat {\bm \theta} = \hat {\bm \theta}( \{ (Y_1,\bm X_1')' , \ldots , (Y_T,\bm X_T')' \} )$.
Notice that clearly $R^G_\mathsf{oos}( {\bm \theta}) = R_\mathsf{oos}( {\bm \theta})$.
We may then note that 
\begin{align*}
 	R_\mathsf{oos}(\hat {\bm \theta}) - R_\mathsf{oos}(\bm \theta^*) 
	& \leq | R_\mathsf{oos}(\hat {\bm \theta}) - R^G_\mathsf{oos}(\hat {\bm \theta})| + | R^G_\mathsf{oos}(\hat {\bm \theta})  - R_\mathsf{oos}(\bm \theta^*) | \\ 
	& = | R_\mathsf{oos}(\hat {\bm \theta}) - R^G_\mathsf{oos}(\hat {\bm \theta})| + | R^G_\mathsf{oos}(\hat {\bm \theta})  - R^G_\mathsf{oos}(\bm \theta^*) | ~.
\end{align*}
The claim of the theorem follows from the fact that if for some $\varepsilon_1 > 0$, $\varepsilon_2>0$, $\delta_1 \in (0,1)$ and $\delta_2 \in (0,1)$  we have that
\begin{align}
	& \mathbb P( | R^G_\mathsf{oos}(\hat {\bm \theta})  - R^G_\mathsf{oos}(\bm \theta^*)| \geq \varepsilon_1 ) \leq \delta_1 \label{eqn:roos:1} \\
	& \mathbb P\left( \left. | R_\mathsf{oos}(\hat {\bm \theta}) - R^G_\mathsf{oos}(\hat {\bm \theta})| \geq \varepsilon_2 \right| | R_\mathsf{oos}^G(\hat {\bm \theta})  - R^G_\mathsf{oos}(\bm \theta^*)| \leq \varepsilon_1 \right) \leq \delta_2 ~, \label{eqn:roos:2} 
\end{align}
then it follows from the union bound and the total probability theorem that \( R_\mathsf{oos}(\hat {\bm \theta}) - R_\mathsf{oos}(\bm \theta^*) \leq \varepsilon_1 + \varepsilon_2 \) with probability at least $1-2\delta_1-\delta_2$. 
Theorem \ref{thm:erm:stationary} implies that for all $T$ sufficiently large \eqref{eqn:roos:1} holds for the choice of $\varepsilon_1$ and $\delta_1$ implied by the theorem. Thus, this proof focuses on establishing that \eqref{eqn:roos:2} holds.
Denote by $\mathcal E = \{ R^G_\mathsf{oos}(\hat {\bm \theta})  - R^G_\mathsf{oos}(\bm \theta^*) \leq 1 \}$
and note that conditional on $\mathcal E$ we have that $1 \geq R^G_\mathsf{oos}(\hat {\bm \theta})  - R^G_\mathsf{oos}(\bm \theta^*) = \| f^G_{\bm \theta^*\,t} - f^G_{\hat {\bm \theta}\,t} \|^2_{L_2} > \underline{\lambda} \| \bm \theta^* - \hat {\bm \theta} \|^2_2 $.
Let $\mathbb E_T(\cdot) = \mathbb E(\cdot \lvert \mathcal I_T)$ be the expectation conditional
on information up to time $T$, with $\mathcal I_T$ the information set at time $T$.
This implies that for $r = (1/\underline{\lambda})^{1\over 2} $ we have
\begin{align*}
	&| \mathbb E_T (Y_{T+h}-f_{\hat {\bm \theta} \,T+h})^2 - \mathbb E_T( Y^G_{T+h}- f^G_{\hat{\bm \theta} \,T+h})^2  | \\
	& \quad \leq \sup_{\bm \theta \in B_2(\bm \theta^*,r)}  | \mathbb E_T (Y_{T+h}-f_{\bm \theta\,T+h})^2 - \mathbb E_T( Y^G_{T+h}-f^G_{\bm \theta \,T+h})^2 | \\
	& \quad = \sup_{\bm \theta \in B_2(\bm \theta^*,r)}  | \mathbb E_T (Y_{T+h}-f_{\bm \theta\,T+h})^2 - \mathbb E( Y_{T+h}-f_{\bm \theta \,T+h})^2 | \\
	& \quad \leq | \mathbb E_T (Y_{T+h}-f_{\bm \theta^*\,T+h})^2  - \mathbb E (Y_{T+h}-f_{\bm \theta^*\,T+h})^2 | + \sup_{\bm v \in B_2(\bm 0,r)} \bm v' [ \mathbb E_T ( \bm X_{T+h} \bm X_{T+h}' )  - \mathbb E( \bm X_{T+h} \bm X_{T+h}' ) ] \bm v \\
	& \quad + 2\sup_{\bm v \in B_2(\bm 0,r)} | [ \mathbb E_T( (Y_{T+h}-f_{\bm \theta^*\,T+h}) \bm X_{T+h} ) - \mathbb E( (Y_{T+h}-f_{\bm \theta^*\,T+h}) \bm X_{T+h} ) ]' \bm v | ~.
\end{align*}
It follows from Ibragimov's inequality that
\begin{align*}
& \| \mathbb E_T (Y_{T+h}-f_{\bm \theta^*\,T+h})^2  - \mathbb E (Y_{T+h}-f_{\bm \theta^*\,T+h})^2 \|_{L_1} \leq 6 \alpha(h)^{1\over 2} \| (Y_{T+h}-f_{\bm \theta^*\,T+h})^2 \|_{L_2} \leq 6 \alpha(h)^{1\over 2} K^2_m ~, \\
& \| \sup_{\bm v \in B_2(\bm 0, r)} \bm v' [ \mathbb E_T ( \bm X_{T+h} \bm X_{T+h}' )  - \mathbb E( \bm X_{T+h} \bm X_{T+h}' ) ] \bm v  \|_{L_1} \\
& \quad \leq \| \max_{ij} | \left[ \mathbb E_T (\bm{X}_{T+h} \bm{X}_{T + h}') - \mathbb E (\bm{X}_{T + h} \bm{X}_{T + h}') \right]_{ij} | \|_{L_1} \sup_{\bm v \in B_2(\bm 0, r)} \| \bm{v} \|_1^2 \leq 6 \alpha(h)^{1\over 2} \frac{K_m^2}{\underline{\lambda}} p ~,\\
& \| 2 \sup_{\bm v \in B_2(\bm 0, r)} | [ \mathbb E_T( (Y_{T+h}-f_{\bm \theta^*\,T+h}) \bm X_{T+h} ) - \mathbb E( (Y_{T+h}-f_{\bm \theta^*\,T+h}) \bm X_{T+h} ) ]' \bm v | \|_{L_1} \\
& \quad \leq 12 \alpha(h)^{1\over 2} \| (Y_{T+h} - f_{\bm \theta^*\,T+h}) X_{i, T+h} \|_{L_2} \sup_{\bm v \in B_2(\bm 0, r)} \| \bm{v} \|_1 \leq 12 \alpha(h)^{\frac{1}{2}} \frac{K_m^2}{\sqrt{\underline{\lambda}}} \sqrt{p} ~.
\end{align*}
Thus, conditional on $\mathcal E$ and for $T$ sufficiently large we have
\begin{align*}
	\| \mathbb E_T (Y_{T+h}-f_{\hat{\bm \theta}\,T+h})^2 - \mathbb E_T( Y^G_{T+h}- f^G_{\hat{\bm \theta} \,T+h})^2 \|_{L_1} 
	\leq 6 \alpha(h)^{1\over 2} K_m^2 \left( 1 + \frac{p}{\underline{\lambda}} + 2 \sqrt{\frac{p}{\underline{ \lambda }}} \right) 
	\leq 24 \alpha(h)^{1\over 2} { K_m^2 \over \underline{\lambda} } p ~.
\end{align*}
The conditional version of Markov's inequality implies that 
\begin{align*}
	\mathbb P( | R_\mathsf{oos}(\hat {\bm \theta}) - R^G_\mathsf{oos}(\hat {\bm \theta})| \geq
    \varepsilon_2 | \mathcal E ) & \leq {1 \over \varepsilon_2 } {1\over H} \sum_{h=1}^{H} \|
    \mathbb E_T (Y_{T+h}-f_{\hat{\bm \theta}\,T+h})^2 - 
    \mathbb E( Y_{T+h}-f_{\hat{\bm \theta} \,T+h})^2 \|_{L_1} \\
	& \leq {24 \over \varepsilon_2} {K_m^2 \over \underline \lambda} \sum_{l=1}^{\infty} \alpha(l)^{1\over 2} {p \over H} ~,
\end{align*}
which implies the claim of the theorem.
\end{proof}

\begin{proof}[Proof of Lemma \ref{lemma:smallball}]
Let $\bm v=\bm \theta_1 - \bm \theta_2$ and note that the Paley-Zygmund inequality implies that for any $\vartheta \in [0,1]$ we have 
\begin{equation}\label{eqn:paleyzygmund}
	\mathbb P( | \bm v' \bm X_t | > \vartheta^{1\over 2} \| \bm v' \bm X_t \|_{L_2} ) 
	\geq (1-\vartheta)^2 {\mathbb E ( | \bm v' \bm X_t |^2 )^2  \over \mathbb E( | \bm v' \bm X_t |^4 ) } ~.
\end{equation}
Note that \ref{asm:spherical} implies that $\bm X_t$ is elliptical. Then $\bm v' \bm X_t = \sigma_t U$ holds  where $\sigma^2_t = \bm v' \bm \Sigma_t \bm v$ and $U$ is an elliptical random variable with zero mean and unit variance (whose distribution does not depend on $\bm v$ nor $\bm \Sigma_t$).
Thus, we have that the probability in \eqref{eqn:paleyzygmund} is lower bounded by \( [ (1-\vartheta)^2 \mathbb E ( | U |^2 )^2  ] / \mathbb E( | U |^4 ) \),
which implies the claim of the lemma.
\end{proof}

\section{Auxiliary Results}

\setcounter{prop}{0}
\setcounter{lemma}{0}
\setcounter{thm}{0}
\renewcommand\theprop{\thesection.\arabic{prop}}
\renewcommand\thelemma{\thesection.\arabic{lemma}}
\renewcommand\thethm{\thesection.\arabic{thm}}

\begin{prop}\label{prop:l1:S}
Consider the same setup as in Proposition \ref{lem:1}.
Let $V_i = \{ \bm v \in \mathbb R^p : \|\bm v-\bm v_i\|_2 \leq \delta \}$ with $\bm v_i \in V$ for $i=1,\ldots,N_\delta$ denote a $\delta$-covering of the set $V$ for some $\delta < {K_{\bm \Sigma}^{-1/2} / ( 2 \overline{\lambda}^{1/2} ) } $.
Define the function $ g_{\bm v\,t} = \mathbbm{1}_{ \{ | \bm X_t' \bm v | \geq \kappa_1 K_{\bm \Sigma}^{-1/2} \} } $ and let $g_{i\,t} = g_{\bm v_i\,t}$.

Then $(i)$ for all $\bm v \in V_i$ we have that \( |g_{\bm v\,t} - g_{i\,t} | \leq \bar{g}_{i\,t} = \mathbbm{1}_{ \{ \bm X_t \in S_{i} \} } \),
where \( S_{i} = \bigcup_{\bm v \in V_i} \{ \bm x \in \mathbb R^p : | \bm x' \bm v | = \kappa_1 K_{\bm \Sigma}^{-1/2} \} \) 
and $(ii)$ there exists a positive constant $K_1$ that depends on $K_{\bm Z}$, $\underline{\lambda}$, $\overline{\lambda}$ and $\kappa_1$ (and it does not depend on $t$, $i$ or $p$) such that \( \mathbb E \bar{g}_{i\,t} \leq K_1 p^{1/2} \delta \).
\end{prop}

\begin{proof}
$(i)$ We show that \( S_{i} = \bigcup_{\bm v \in V_i} \{ \bm x \in \mathbb R^p : | \bm x' \bm v | = \kappa_1 K_{\bm \Sigma}^{-1/2}\} \)
is the set containing all the vectors $\bm x$ such that the indicator functions $ \mathbbm{1}_{ \{ |\bm x' \bm v | \geq \kappa_1 K_{\bm\Sigma}^{-1/2} \} }$ and 
$ \mathbbm{1}_{ \{ | \bm x' \bm v_i | \geq \kappa_1 K_{\bm\Sigma}^{-1/2} \} } $ are different. 
We do so by showing that the complement of $S_{i}$ is a set of vectors $\bm x$ where the indicator functions are equal.
We establish this by contradiction. Assume $\bm x$ is not in $S_i$ and that the indicator functions $ \mathbbm{1}_{ \{| \bm x' \bm v | \geq \kappa_1 K_{\bm\Sigma}^{-1/2} \}}$ 
and $\mathbbm{1}_{\{ | \bm x'\bm v_i | \geq \kappa_1 K_{\bm\Sigma}^{-1/2} \}} $ are different.
Since $V_i$ is convex there must be an intermediate $\dot {\bm v} \in V_i$ such that $| \bm x'\dot{ \bm v} | = \kappa_1 K_{\bm\Sigma}^{-1/2} $ implying that $\bm x$ is in $S_i$, which leads to a contradiction. \\
$(ii)$ Note that
\begin{equation*}
	S_{i} = \bigcup_{\bm v \in V_i} \{ \bm x \in \mathbb R^p : \bm x'\bm v = \kappa_1 K_{\bm \Sigma}^{-1/2} \} \cup \bigcup_{\bm v \in V_i} \{ \bm x \in \mathbb R^p : \bm x'\bm v = -\kappa_1 K_{\bm \Sigma}^{-1/2} \} 
    = S_{i\,+} \cup S_{i\,-} ~. 
\end{equation*}
In what follows we bound the probability of the event $\{ \bm X_t \in S_{i\,+} \}$ only as the event $\{ \bm X_t \in S_{i\,-} \}$ can be treated analogously. We divide the proof into four steps. \\
$\bf 1.$ We work with an appropriately rotated version of $\bm X_t$ denote by $\bm Z$.
Let $\vartheta$ be the angle between the vector $\bm \Sigma^{1/2}_t \bm v_i$ and $(1,0,\ldots,0)'$ and let $\mathbf R \in \mathbb R^{p \times p} $ be the rotation matrix associated with $\vartheta$.
Recall that: 
$(i)$ $\mathbf R' \mathbf R = \mathbf I_p$;
$(ii)$ $\mathbf R \bm \Sigma^{1/2}_t \bm v_i = \|\bm \Sigma^{1/2}_t \bm v_i\|_2 (1,0,\ldots,0)'$;
$(iii)$ if we define $W_1 = \{ \bm w \in \mathbb R^p : \| \bm w \|_2 \leq 1 \}$ and $W_{2} = \{ \bm w \in \mathbb R^p : \bm w = \mathbf R \bm w^\star \text{ for some } \bm w^\star \in W_1 \}$ then we have that $W_1 = W_2$.
Define $\bm Z = \mathbf R \bm \Sigma^{-1/2}_t \bm X_t$ and note that
\begin{align*}
	\mathbb P( \{ \bm X_t \in S_{i\,+} \} )
	&= \mathbb P\left( \left\{ \bm X_t \in \bigcup_{\bm w \in \mathbb R^p: \| \bm w \|_2 \leq 1 } \{ \bm x \in \mathbb R^p : \bm v_i'\bm x + \delta \bm w'\bm x = \kappa_1 K_{\bm \Sigma}^{-1/2} \} \right\} \right) \\
	&= \mathbb P\left( \left\{ \bm Z \in \bigcup_{\bm w \in \mathbb R^p: \| \bm w \|_2 \leq 1} \{ \bm z \in \mathbb R^p : \bm v_i' \bm \Sigma^{1/2}_t \mathbf R' \bm z + \delta \bm w'\bm \Sigma^{1/2}_t \mathbf R' \bm z = \kappa_1 K_{\bm \Sigma}^{-1/2} \} \right\} \right) ~.
\end{align*}
Define $c_{i\,t} = \| \bm \Sigma^{1/2}_t \bm v_i \|_2 $ and note that the set in the last equation is such that
\begin{align*}
	&\bigcup_{\bm w \in \mathbb R^p: \| \bm w \|_2 \leq 1} \{ \bm z \in \mathbb R^p :  \| \bm \Sigma^{1/2}_t \bm v_i  \|_2 z_1 + \delta (\mathbf R \bm \Sigma^{1/2}_t \bm w)'\bm z = \kappa_1 K_{\bm \Sigma}^{-1/2} \} \\
	&\quad \subset \bigcup_{\bm w \in \mathbb R^p: \| \bm w \|_2 \leq 1} \{ \bm z\in \mathbb R^p :  c_{i\,t} z_1 + \overline{\lambda}^{1/2} \delta \bm w'\bm z = \kappa_1 K_{\bm \Sigma}^{-1/2} \}  \\
	&\quad = \bigcup_{\bm w \in \mathbb R^p: \| \bm w \|_2 \leq 1} \{ \bm z \in \mathbb R^p :  ( c_{i\,t} + \overline{\lambda}^{1/2}\delta w_1) z_1 + \overline{\lambda}^{1/2} \delta \bm w_{-1}'\bm z_{-1} = \kappa_1 K_{\bm \Sigma}^{-1/2}  \} 
	= S'_{it\,+} ~.
\end{align*}
Lastly, we note that for any $i=1,\ldots,N_\delta$ and $t=1,\ldots,T$ we have 
\[ 
	c_{i\,t} 
	= { \| \bm \Sigma^{1/2}_t (\bm \theta^*-\bm \theta_i) \|_2 \over \sqrt{ {1 \over T }\sum_{t=1}^T \| \bm X_t'(\bm \theta^* - \bm \theta_i) \|^2_{L_2}  } } 
	= \sqrt{ (\bm \theta^*-\bm \theta_i)' \bm \Sigma_t  (\bm \theta^*-\bm \theta_i) \over {1 \over T }\sum_{t=1}^T (\bm \theta^* - \bm \theta_i)' \bm \Sigma_t   (\bm \theta^* - \bm \theta_i) } > K_{\bm \Sigma}^{-1/2} ~, 
\] 
and  $c_{i\,t} - \overline{\lambda}^{1/2} \delta > K_{\bm \Sigma}^{-1/2} / 2 > 0$.\\
$\bf 2.$ We construct two sets $S_{it\,+}^{'1}$ and $S_{it\,+}^{'2}$ such that $S_{it\,+}' \subset S_{it\,+}^{'1} \cap S_{it\,+}^{'2}$. 
Define 
\begin{equation}\label{eqn:cone1}
	S_{it\,+}^{'1} = \left\{ \bm z \in \mathbb R^p : z_1 \leq 
	{ \kappa_1 K_{\bm \Sigma}^{-1/2} \over c_{i\,t} - \overline{\lambda}^{1/2} \delta } + {\overline{\lambda}^{1/2} \delta \over c_{i\,t} - \overline{\lambda}^{1/2}\delta } \sqrt{ z_{2}^2 + \ldots + z_{p}^2 } \right\} ~,
\end{equation}
that is the set of points ``underneath'' a hyper-cone.
Let $\bm z$ be in $S_{it\,+}' $, define $\dot {\bm z} = \|\bm z_{-1}\|^{-1}_2 (z_2,\ldots,z_{p})'$ and note that $\|\dot {\bm z}\|_2=1$. Then for some $\bm w$ such that $\| \bm w \|_2 \leq 1$ we have that
\begin{align*}
	z_1 &=  { \kappa_1 K_{\bm \Sigma}^{-1/2} \over c_{i\,t} + \overline{\lambda}^{1/2}\delta w_1 } - { \overline{\lambda}^{1/2}\delta \bm w_{-1}'\bm z_{-1} \over c_{i\,t} + \overline{\lambda}^{1/2}\delta w_1 } 
	= { \kappa_1 K_{\bm \Sigma}^{-1/2} \over c_{i\,t} + \overline{\lambda}^{1/2}\delta w_1 } - { \overline{\lambda}^{1/2}\delta \bm w_{-1}'\dot {\bm z} \over c_{i\,t} + \overline{\lambda}^{1/2}\delta w_1 } \| \bm z_{-1} \|_2 \\
	&\leq { \kappa_1 K_{\bm \Sigma}^{-1/2} \over c_{i\,t} - \overline{\lambda}^{1/2} \delta } + {\overline{\lambda}^{1/2} \delta \| \bm w_{-1}\|_2 \|\dot {\bm z}\|_2 \over K_{\bm \Sigma}^{-1/2}- \overline{\lambda}^{1/2}\delta } \| \bm z_{-1} \|_2 
	\leq { \kappa_1 K_{\bm \Sigma}^{-1/2} \over c_{i\,t} - \overline{\lambda}^{1/2} \delta } + {\overline{\lambda}^{1/2} \delta \over c_{i\,t} - \overline{\lambda}^{1/2}\delta } \sqrt{ z_{2}^2 + \ldots + z_{p}^2 } ~, 
\end{align*}
which implies that $\bm z$ is also in $S_{i\,+}^{'1}$.
Define  
\begin{equation}\label{eqn:cone2}
	S_{it\,+}^{'2} = \left\{ \bm z \in \mathbb R^p :  z_1 \geq { \kappa_1 K_{\bm \Sigma}^{-1/2} \over c_{i\,t} + \overline{\lambda}^{1/2} \delta } - {\overline{\lambda}^{1/2} \delta \over c_{i\,t} - \overline{\lambda}^{1/2}\delta } \sqrt{ z_{2}^2 + \ldots + z_{p}^2 } \right\} ~,
\end{equation}
that is the set of points ``above'' a hyper-cone. 
Let $\bm z$ in $S_{it\,+}' $ and define $\dot{\bm z}$ as above. Then for some $\bm w$ such that $\| \bm w \|_2 \leq 1$ we have that
\begin{align*}
	z_1 &= { \kappa_1 K_{\bm \Sigma}^{-1/2} \over c_{i\,t} + \overline{\lambda}^{1/2}\delta w_1 } - { \overline{\lambda}^{1/2}\delta \bm w_{-1}'\bm z_{-1} \over c_{i\,t} + \overline{\lambda}^{1/2}\delta w_1 } 
	= { \kappa_1 K_{\bm \Sigma}^{-1/2} \over c_{i\,t} + \overline{\lambda}^{1/2}\delta w_1 } - { \overline{\lambda}^{1/2}\delta \bm w_{-1}'\dot {\bm z} \over c_{i\,t} + \overline{\lambda}^{1/2}\delta w_1 } \| \bm z_{-1} \|_2 \\
	&\geq { \kappa_1 K_{\bm \Sigma}^{-1/2} \over c_{i\,t} + \overline{\lambda}^{1/2} \delta } - {\overline{\lambda}^{1/2} \delta \| \bm w_{-1}\|_2 \|\dot {\bm z}\|_2 \over K_{\bm \Sigma}^{-1/2}- \overline{\lambda}^{1/2}\delta } \| \bm z_{-1} \|_2  
	\geq { \kappa_1 K_{\bm \Sigma}^{-1/2} \over c_{i\,t} + \overline{\lambda}^{1/2} \delta } - {\overline{\lambda}^{1/2} \delta \over c_{i\,t} - \overline{\lambda}^{1/2}\delta } \sqrt{ z_{2}^2 + \ldots + z_{p}^2 } ~,
\end{align*}
which implies that $\bm z$ is also in $S_{i}^{+,2}$. \\
$\bf 3.$ We establish an upper bound on the probability of the event $\{ \bm Z \in S_{it\,+}' \}$. 
Note that $ S_{it\,+}' \subset S_{i\,+}^{'1} \cap S_{i\,+}^{'2} = A_i \cup B_i \cup C_i $ where
\begin{align*}
	A_{it} &= S_{it\,+}^{'1} \cap \left\{\bm z \in \mathbb R^p : z_1 \geq { \kappa_1 K_{\bm \Sigma}^{-1/2} \over c_{i\,t} - \overline{\lambda}^{-1/2} \delta } \right\}  \\
	B_{it} &= \left\{\bm z \in \mathbb R^p : { \kappa_1 K_{\bm \Sigma}^{-1/2} \over c_{i\,t} + \overline{\lambda}^{1/2} \delta } \leq z_1 \leq { \kappa_1 K_{\bm \Sigma}^{-1/2} \over c_{i\,t} - \overline{\lambda}^{1/2} \delta } \right\} \\
	C_{it} &= S_{it\,+}^{'2} \cap \left\{\bm z \in \mathbb R^p : z_1 < { \kappa_1 K_{\bm \Sigma}^{-1/2} \over c_{i\,t} + \overline{\lambda}^{1/2} \delta } \right\}  ~.
\end{align*}
Then we have that $\mathbb P( \bm X_t \in S_{i\,+}) < \mathbb P( \bm Z \in A_{i\,t} ) + \mathbb P( \bm Z \in B_{i\,t} )  + \mathbb P( \bm Z \in C_{i\,t} ) $.
Using Proposition \ref{prop:nspherical} and \ref{asm:distribution} we have that
\begin{align*}
	\mathbb P( \bm Z \in A_{it} ) &\leq K_{\bm Z} K_{\bm \Sigma}^{1/2} \overline{\lambda}^{1/2} \sqrt{\pi \over 2} p^{1/2} \delta \\
	\mathbb P( \bm Z \in B_{it} ) 
	&= \mathbb P\left( {\kappa_1 K_{\bm \Sigma}^{-1/2} \over c_{i\,t} + \overline{\lambda}^{1/2}  \delta } \leq Z_{1} \leq { \kappa_1 K_{\bm \Sigma}^{-1/2} \over c_{i\,t} - \overline{\lambda}^{1/2} \delta } \right) \\
	&\leq K_{\bm Z} \sup_{s} f_{S_{1}}(s) \kappa_1 K_{\bm \Sigma}^{-1/2} { 2 \overline{\lambda}^{1/2} \delta \over (c_{it} - \overline{\lambda}^{1/2} \delta) (c_{it} + \overline{\lambda}^{1/2} \delta) } 
	\leq 8 K_{\bm Z} \kappa_1 K_{\bm \Sigma}^{1/2} \sup_{s} f_{S_{1}}(s)  \overline{\lambda}^{1/2} \delta \\
	\mathbb P( \bm Z \in C_{it} ) &\leq K_{\bm Z} K_{\bm \Sigma}^{1/2} \overline{\lambda}^{1/2} \sqrt{\pi \over 2} p^{1/2} \delta ~.
\end{align*}
$\bf 4.$ It follows from the inequalities above, and by using analogous steps to bound the probability of the event $\mathbb P( \bm X_t \in S_{i\,-})$, that there exists a positive constant $K_1$ that depends on $K_{\bm Z}$, $\bm S$, $\underline{\lambda}$, $\overline{\lambda}$ and $\kappa_1$, but does not depend on $i$ and $t$ or $p$, such that \(\mathbb P( \bm X_t \in S_{i} ) \leq K_1 p^{1/2} \delta \).
\end{proof}

\begin{prop}\label{prop:nspherical}

Let $\bm Z$ be a $p$-dimensional random vector.
Suppose $\mathbb P( \bm Z \in E ) \leq K_{\bm Z} \mathbb P( \bm S \in E ) $ holds 
for some $p$-dimensional spherical random vector $\bm S$ whose density is assumed to exist, 
some positive constant $K_{\bm Z}$ and any $E \in \mathcal B(\mathbb R^p)$.
Define the set \( S = \{ \bm z \in \mathbb R^p : a \leq z_1 \leq a + b \sqrt{ z_2^2 + \ldots + z_{p}^2 } \} \) for some $a,b>0$.

Then, there is a positive constant $C$ such that \( \mathbb P( \bm Z \in S ) \leq C p^{1\over 2} b \).
\end{prop}

\begin{proof}
For convenience we show this result for $p>2$ and for $a=0$. We have
\begin{align*}
	& \mathbb P\left( \bm Z \in S \right) = \mathbb P\left( 0 \leq Z_1 \leq b \sqrt{ Z_2^2 + \ldots
    + Z_p^2 } \right) = \mathbb P\left( 0 \leq Z_1^2 \leq b^2 ( Z_2^2 + \ldots + Z_p^2 ) \right)\\
	& \quad = \mathbb P\left( 0 \leq { Z_1^2 \over \| \bm Z \|^2_{2} } \leq b^2 { \| \bm Z \|^2_{2} - Z^2_1 \over \| \bm Z \|^2_{2} } \right) 
	= \mathbb P\left( 0 \leq {Z_1 \over \| \bm Z \|_2 } \leq { b \over \sqrt{1 + b^2} }\right) ~.
\end{align*}
Consider the $p$-spherical transformation of $\bm Z$ \citep[Example 1.6.8]{Fang:Zhang:1990}
\begin{equation*}
	(Z_1,\ldots,Z_i,\ldots,Z_p)' = r \left( \cos \theta_1 , \ldots , \prod_{k=1}^{i-1} \sin \theta_k \cos \theta_i, \ldots, \prod_{k=1}^{p-2} \sin \theta_k \sin \theta_{p-1} \right)'~,
\end{equation*}
where $r \in [0, \infty)$, $\theta_i \in [0, \pi]$ for $1 \leq i \leq p-2$ and $\theta_{p-1} \in [0,2\pi]$. 
We remark that $r$ denotes $\| \bm Z \|_2$ and that
the angles $\theta_1, \ldots, \theta_{p-1}$ are set according to the following scheme:
$\theta_1$ is the angle between the $z_1$ axis and the vector $\bm Z$;
$\theta_2$ is the angle between the projection of the $\bm Z$ vector on the span generated by $z_2,\ldots,z_{p}$, which we denote by $\bm Z^{(1)}$, and the $z_2$ axis;
$\theta_3$ is the angle between the projection of $\bm Z^{(1)}$ on the span generated by $z_3,\ldots,z_{p}$, which we denote by $\bm Z^{(2)}$, and the $z_3$ axis;
$\ldots$;
$\theta_{p-1}$ is the angle between the projection of $\bm Z^{(p-2)}$ on the span generated by $Z_{p-1},Z_{p}$ and the $z_{p-1}$ axis.
If we let $\vartheta$ denote the angle such that $\cos( \vartheta ) = b/\sqrt{1+b^2}$ then we have
\begin{eqnarray}
	&& \mathbb P\left( 0 \leq {Z_1 \over \| \bm Z \|_2 } \leq { b \over \sqrt{1 + b^2} }\right)
	= \mathbb P\left( 0 \leq \cos \theta_1 \leq { b \over \sqrt{1 + b^2} }\right) = \mathbb P\left( \vartheta \leq \theta_1 \leq {\pi \over 2} \right) ~. \label{eqn:ptheta1}
\end{eqnarray}
We use $K_{\bm Z}$ and the distribution of $\bm S$ to bound the probability in \eqref{eqn:ptheta1}.
\citet[Theorem 2.11]{Fang:Kotz:Ng:1990} establishes that the density of $\theta_1$ implied by $\bm S$ is given \( f_{\theta_1}(t) = {\Gamma\left({p\over2}\right)\over \Gamma\left({1\over2}\right) \Gamma\left({p-1\over2}\right)} \sin^{p-2} t \).
Then we have that \eqref{eqn:ptheta1} is upper bounded by
\begin{eqnarray*}
	K_{\bm Z} \int_{\vartheta}^{\pi/2} f_{\theta_1}(t) dt 
	\stackrel{(a)}{\leq} {K_{\bm Z} \over \sqrt{\pi} } {\Gamma\left({p\over2}\right) \over \Gamma\left({p-1\over2}\right)}  \int_{\vartheta}^{\pi/2} 1 dt 
	\stackrel{(b)}{\leq} {K_{\bm Z} \over \sqrt{2\pi} } p^{1/2} \left( {\pi\over 2} - \vartheta \right) 
	\stackrel{(c)}{\leq} {K_{\bm Z} \over 2} \sqrt{\pi \over 2} p^{1/2} {b \over \sqrt{1+b^2}} ~, 
\end{eqnarray*}
where 
$(a)$ follows from the fact that for any $\theta$ it holds that $\sin^{p-2} \theta \leq 1$,
$(b)$ follows from  the fact that for $x>0$ and $s\in (0,1)$ it holds that $ \Gamma(x + 1) / \Gamma(x + s) < (x + 1)^{1 - s}$ (Gautschi's inequality) and 
$(c)$ follows from the fact that \( {\pi \over 2} - \vartheta  = {\pi \over 2} - \arccos( \cos(\vartheta) ) \leq {\pi \over 2} \cos(\vartheta)  \).
The claim then follows since for any $b>0$ we have that $b/\sqrt{1+b^2} < b$.
\end{proof}

\begin{prop}\label{prop:l1:conc}
Let $\{ Z_t \}_{t=1}^T$ be a sequence of centered Bernoulli random variables.
Suppose that the $\alpha$-mixing coefficients of the sequence satisfy $ \alpha(l) < \exp( -K_\alpha l^{r_\alpha} ) $ for some $K_\alpha>0$ and $r_\alpha > 0$.

Define $p = \lfloor K_p T^{r_p} \rfloor$ for some $K_p>0$ and $r_p \in [0, r_\alpha/(r_\alpha+1) )$ and define 
\[
 	\varepsilon_T = { \sqrt{ {K_1 K_2 }{ p \log (T) } \over T^{  r_\alpha \over r_\alpha+1} } } + \sqrt{ K_2 {\log (T) } \over T^{  r_\alpha \over r_\alpha+1} }  ~, 
\]
where $K_1= { 3/ 2} $ and $K_2 = 64 \sigma^2$ with $\sigma^2 = ( {1\over 4}+ 8 \sum_{l=1}^\infty \alpha(l))$.

Then, for any $K_3>0$ and all $T$ sufficiently large, it holds that 
\[
	\left( 1 + {K_3 p^{1\over 2} \over \varepsilon_T } \right)^p \mathbb P\left( \left| {1\over T} \sum_{t=1}^T Z_t \right| > \varepsilon_T \right) \leq {4 \over T} + o\left( { 1 \over T } \right) ~. 
\]
\end{prop}

\begin{proof}
We begin by noting that $Z_t$ is zero-mean and that $\sup_{1 \leq t \leq T} \| Z_t \|_{L_\infty} \leq 1$, thus it satisfies the mixing and moment conditions of Theorem 2.1 of \citet{Liebscher:1996}.
Define $M_T = \lfloor T^{ 1 \over r_\alpha+1} \rfloor$ and note that for all $T \geq 2$ 
we have that $ M_T \in [1,T] $ and \( 4 M_T <  T \varepsilon_T \), as required by the theorem.
Then, we have
\begin{equation*}
\mathbb P\left( \left| \sum_{t=1}^T Z_t \right| > T \varepsilon_T \right) \leq 4 \exp\left( - { T \varepsilon_T^2 \over 64 D(T,M_T) / M_T + {8 \over 3} M_T \varepsilon_T  } \right) + 4 { T \over M_T } e^{-K_\alpha M_T^{r_\alpha}} ~,
\end{equation*}
with \( D(T,M_T) = \sup_{0\leq j\leq T-1} \mathbb E\left[ \left( \sum_{t=j+1}^{j+M_T \wedge T} Z_{t} \right)^2 \right] \).
Define $\gamma(l) = \sup_{1\leq t \leq T-l} |\operatorname{Cov}(Z_{t},Z_{t+l})|$ for $l=0,\ldots,T-1$.
Note that $\gamma(0) \leq 1/4$ and, by Billingsley's inequality \citep[Corollary 1.1]{Bosq:1998}, that $\gamma(l) \leq 4 \alpha(l)$ for $l=1,\ldots,T-1$.
Thus, it holds that $D(T,M_T) \leq M_T \gamma(0) + 2 M_T \sum_{l=1}^{M_T-1} \gamma(l) \leq M_T ( {1\over 4} + 8 \sum_{l=1}^{\infty} \alpha(l) ) = M_T \sigma^2 $.
We have
\begin{align*}
& \left( 1 + {K_3 p^{1\over 2} \over \varepsilon_T} \right)^p \mathbb P\left( \left| {1 \over T} \sum_{t=1}^T Z_t \right| > \varepsilon_T \right) \\
& \quad \leq 4 \left( 1 + {K_3 p^{1\over 2}  \over \varepsilon_T } \right)^p\exp\left( - { T \varepsilon_T^2 \over {64 \sigma^2} + {8 \over 3} M_T \varepsilon_T  } \right) + 4 \left( 1 + {K_3 p^{1\over 2} \over \varepsilon_T} \right)^p{ T \over M_T } e^{-K_\alpha M_T^{r_\alpha}} \\
& \quad \leq 4 \left( 1 + {K_3 p^{1\over 2} \over \varepsilon_T } \right)^p\exp\left( - { T^{r_\alpha \over r_\alpha+1} \varepsilon_T^2 \over {64 \sigma^2} + {8\over 3} \varepsilon_T  } \right) 
+ 8 \left( 1 + {K_3 p^{1\over 2} \over \varepsilon_T} \right)^p T^{r_\alpha \over r_\alpha+1 } \exp\left( - { K_\alpha  \over 2^{r_\alpha} } T^{r_\alpha \over r_\alpha+1 } \right)  \\
& \quad \stackrel{(a)}{\leq} 4 \left( 1 + {K_3 p^{1\over 2} \over \varepsilon_T} \right)^p\exp\left( - { T^{r_\alpha \over r_\alpha + 1} \over	64 \sigma^2 } \varepsilon_T^2 \right) 
+ 8 \left( 1 + {K_3 p^{1\over 2} \over \varepsilon_T} \right)^p T^{r_\alpha \over r_\alpha+1 } \exp\left( - { K_\alpha  \over 2^{r_\alpha} } T^{r_\alpha \over r_\alpha+1 } \right)  \\
& \quad = A_T + B_T ~,
\end{align*}
where $(a)$ follows from the fact that $x^2/(64 \sigma^2+{8\over 3}x) \geq x^2/(64 \sigma^2)$ for any $x>0$ and $\sigma^2 > 1/4$.
Note that for all $T$ sufficiently large we have
\begin{align*}
	\log \left( 1 + {K_3 p^{1\over 2} \over \varepsilon_T } \right)^p 
	& = p \log\left( \varepsilon_T + K_3 p^{1\over 2} \right) - p \log (\varepsilon_T) \\
	& = {1\over 2} p \log (p) + p \log\left( {\varepsilon_T \over p^{1\over 2}} + K_3 \right) - p \log (\varepsilon_T) \\
	& \leq {1\over 2} p \log (K_p) + {r_p\over 2} p \log (T) + p \log\left( 1 + K_3 \right) - p \log (\varepsilon_T) \\
	& \leq \left( {1\over 2} + {r_p\over 2} + 1 + {r_\alpha \over 2(r_\alpha+1)} \right) p \log (T) < K_1 p \log (T) ~,
\end{align*}
where the last inequality follows from the fact that $r_p<1$ and ${r_\alpha / (2(r_\alpha+1))} < 1$.
Finally, the claim follows after noting that for all $T$ sufficiently large we have 
\begin{align*}
	& A_T \leq 4 \exp\left( K_1 p \log (T) - { T^{r_\alpha \over r_\alpha+1} \over K_2 } \varepsilon^2_T \right) \\
	& \quad \stackrel{(a)}{\leq} 4 \exp \left( K_1 p \log (T) - K_1 p \log (T) - \log (T) \right) \leq 4 \exp \left( - \log (T) \right) = {4\over T} ~,
\end{align*}
where $(a)$ follows from the fact that $(x+y)^2 \geq x^2 + y^2$ for $x,y\geq0$, and that
\begin{align*}
	B_T & \leq 8 \left( 1 + {K_3 p^{1/2} \over \varepsilon_T} \right)^p T^{  r_\alpha  \over r_\alpha + 1 } \exp\left( - { K_\alpha  \over 2^{r_\alpha} } T^{r_\alpha \over r_\alpha+1 } \right)  \\
	& \leq 8 \exp \left( {r_\alpha \over r_\alpha +1 } \log (T) + K_1 K_p T^{r_p} \log (T) - {K_\alpha \over 2^{r_\alpha}} T^{r_\alpha \over r_\alpha + 1} \right) = o\left({1 \over T}\right)~.
\end{align*}
\end{proof}

\begin{prop}\label{prop:l2:conc}
Let $\{\bm Z_t \}_{t=1}^T$ be a sequence of $p$-dimensional zero-mean random vectors.
Suppose that
(i) $\sup_{1 \leq i \leq p} \sup_{1\leq t \leq T} \| Z_{i\,t} \|_{L_{2}} \leq K_m $ and  $\sup_{1 \leq i \leq p} \sup_{1\leq t \leq T} \| Z_{i\,t} \|_{L_{r_m}} \leq 2K_m $ for some $K_m \geq 1$ and $r_m > 2$;
(ii) the $\alpha$-mixing coefficients of the sequence satisfy $ \alpha(l) < \exp( -K_\alpha l^{r_\alpha} ) $ for some $K_\alpha>0$ and $r_\alpha > 0$; and
(iii) $p = \lfloor K_p T^{r_p} \rfloor$ for some $K_p>0$ and $r_p \in [0, (r_m-2)/2 \wedge 1)$.

Then, for all $T$ sufficiently large it holds that
\[
	\mathbb P \left(
	\left\| {1 \over T} \sum_{t=1}^T \bm Z_t \right\|_2 > 12 \, \sigma \, \sqrt{ {p \log (T) \over T } }   
	\right) \leq {3 K_p (2K_m)^{r_m} \over \sigma \log (T) } + o\left({ 1 \over \log (T) }\right) ~,
\]
where $\sigma^2 = K_m^2 ( 1 + 128 {r_m \over r_m -2} \sum_{l=1}^\infty \alpha(l)^{1-{2\over r_m} } )$. 
\end{prop}

\begin{proof}
For any positive constant $K$ we have that
\begin{align*}
	\mathbb P\left( \left\| {1 \over T} \sum_{t=1}^T \bm Z_t \right\|_2 
    > K \sqrt{ {p \log (T) \over T } } \right) 
	& \leq \mathbb{P} \left( \max_{1 \leq i \leq p} \left| {1 \over T} \sum^T_{t=1} Z_{it}\right| 
    > K \sqrt{ \log (T) \over T } \right)  \\
	& \leq p \max_{1 \leq i \leq p} \mathbb P\left( \left| \sum_{t=1}^T Z_{i\,t} \right| 
    > K \sqrt{T \log (T) } \right) ~.
\end{align*}
Let \( \sum_{t=1}^T Z_{i\,t} = \sum_{t=1}^T Z'_{i\,t} + \sum_{t=1}^T Z''_{i\,t} \) where \( Z'_{i\,t} =  Z_{i\,t} \mathbbm{1}(|Z_{i\,t}| \leq b_T) - \mathbb E \left( Z_{i\,t} \mathbbm {1}(|Z_{i\,t}| \leq b_T ) \right) \) and 
\( Z''_{i\,t}  =  Z_{i\,t} \mathbbm{1}(|Z_{i\,t}| > b_T)    - \mathbb E \left( Z_{i\,t} \mathbbm {1}(|Z_{i\,t}| > b_T) \right) \).
For any $\lambda \in (0,1)$ we have
\begin{align*}
	& p \max_{1 \leq i \leq p} \mathbb P\left( \left|  \sum_{t=1}^T Z_{i\,t} \right| 
    > K \sqrt{T \log (T)} \right) \\
	& \quad \leq p \max_{1 \leq i \leq p} \mathbb P\left( \left|\sum_{t=1}^T Z'_{i\,t} \right| 
    > \lambda K \sqrt{T \log (T)} \right) 
    + p \max_{1 \leq i \leq p} \mathbb P\left( \left| \sum_{t=1}^T  Z''_{i\,t} \right| 
    > (1-\lambda) K \sqrt{T \log (T) } \right) ~.
\end{align*}
The sequence $\{ Z'_{i\,t} \}_{t=1}^T$ has the same mixing properties as $\{ \bm Z_{t} \}_{t=1}^T$ and $\sup_{1 \leq i \leq p} \sup_{1\leq t\leq T} \| Z_{i\,t}' \|_\infty < 2b_T$. 
Define
\( \varepsilon_T' = \lambda K T^{1\over 2} \sqrt{\log (T) } \),
\( b_T = ( T^{1+2 r_p\over 2} \sqrt{\log (T)})^{1\over r_m-1} \) and 
\( M_T = \lfloor b_T^{-1} T^{ {1\over 2} } / \sqrt{\log (T)} \rfloor  \).
For all $T$ sufficiently large the conditions of Theorem 2.1 of \citet{Liebscher:1996} are satisfied, since for all $T$ sufficiently large we have that $ M_T \in [1,T] $ and \( 4 (2 b_T) M_T < \varepsilon_T' \)
and we have
\begin{equation*}
	p \max_{1 \leq i \leq p} \mathbb P\left( \left| \sum_{t=1}^T Z'_{i\,t}\right| > \varepsilon_T' \right) 
	\leq 4 p \exp\left( - { (\varepsilon_T')^2 \over 64 {T \over M_T}  D(T,M_T) + {16 \over 3} b_T M_T \varepsilon_T' } \right) + 4 {pT\over M_T} \exp \left( -K_\alpha M_T^{r_\alpha} \right) ~,
\end{equation*}
with \( D(T,M_T) = \sup_{0\leq j\leq T-1} \mathbb E\left[ \left( \sum_{t=j+1}^{j+M_T \wedge T} Z'_{i\,t} \right)^2 \right] \).
Define $\gamma(l) = \sup_{1\leq i\leq p} \sup_{1\leq t \leq T-l} |\operatorname{Cov}(Z'_{i\,t}, Z'_{i\,t+l})|$ for $l=0,\ldots,T-1$
and note that $ D(T,M_T) \leq M_T \sum_{l=-T+1}^{T-1} \gamma(l)$. 
Next, we note that $\gamma(0) \leq K_{m}^2$ since 
\begin{eqnarray*}
 	\operatorname{Var}( Z'_{i\,t} ) = \| Z_{i\,t} \mathbbm{1}(|Z_{i\,t}| \leq b_T) \|^2_{L_2} - [ \mathbb E( Z_{i\,t} \mathbbm{1}(|Z_{i\,t}| \leq b_T) ) ]^2 \leq \| Z_{i\,t} \|^2_{L_2} \leq K_{m}^2 ~.
\end{eqnarray*}
Davydov's inequality \citep[Corollary 1.1]{Bosq:1998} implies that 
\begin{eqnarray*}
	\gamma(l) & \leq & 4 {r_m\over r_m-2} \alpha(l)^{1 - {2\over r_m}} \| Z'_{i\,t} \|_{L_{r_m}} \| Z'_{i\,t+l} \|_{L_{r_m}} \leq 64 K_m^2  {r_m\over r_m-2} \alpha(l)^{1 - {2\over r_m}} ~, 
\end{eqnarray*}
for $l=1,\ldots,T-1$, where we have used the fact that 
\begin{eqnarray*}
	\| Z_{i\,t}' \|_{L_{r_m}} 
	& \leq & \| Z_{i\,t} \mathbbm{1}(|Z_{i\,t}| \leq b_T) \|_{L_{r_m}} + \|\mathbb E( Z_{i\,t} \mathbbm{1}(|Z_{i\,t}| \leq b_T)  ) \|_{L_{r_m}}  \\
	& \leq & \| Z_{i\,t} \mathbbm{1}(|Z_{i\,t}| \leq b_T) \|_{L_{r_m}} + \|  Z_{i\,t} \mathbbm{1}(|Z_{i\,t}| \leq b_T)  \|_{L_1}  \\
	& \leq & 2 \| Z_{i\,t} \mathbbm{1}(|Z_{i\,t}| \leq b_T) \|_{L_{r_m}} \leq 2 \| Z_{i\,t}  \|_{L_{r_m}} \leq 4 K_m ~.
\end{eqnarray*}
These together imply that $D(T,M_T) \leq M_T K_m^2 (1 + 128 {r_m \over r_m -2} \sum_{l=1}^\infty \alpha(l)^{1-{2\over r_m} }) = M_T \sigma^2 $.
For any $K$ that satisfies 
\begin{equation} \label{eqn:k1cond}
	K  > {1\over \lambda} \left( 8 \sqrt{ \sigma^2 + {1\over 9} } + {8\over 3} \right)~,
\end{equation}
we have \( 1 - { \lambda^2 K^2  / ( 64 \sigma^2 + {16\over 3} \lambda K } ) < 0 \). 
Notice that the condition is satisfied, for instance, by $K = \lambda^{-1} 8 \sqrt{ 2 \sigma^2 }$ since $\sigma^2 \geq 1$.
Thus, for any $K$ that satisfies this we have
\begin{align*}
	p \max_{1 \leq i \leq p} \mathbb P\left( \left| \sum_{t=1}^T Z'_{i\,t}\right| > \varepsilon_T' \right) 
	& \leq 4 K_p \exp\left( r_p \log (T) - { \lambda^2 K^2 T \log (T)  \over 64 \sigma^2 T + {16\over 3}\lambda K T } \right) + 4 K_p T^{1+r_p} \exp\left( - K_\alpha M_T^{r_\alpha} \right) \\
	& \leq 4 K_p \exp\left( \left[ r_p - { \lambda^2 K^2 \over 64 \sigma^2 + {16\over 3} \lambda K } \right] \log (T)  \right) + 4 K_p T^{1+r_p} \exp\left( - K_\alpha M_T^{r_\alpha} \right) \\
	& \leq o\left({1\over \log (T)}\right)  ~.
\end{align*}
Let $\varepsilon_T'' = (1-\lambda) K T^{1\over 2} \sqrt{\log (T)}$ and note that
\begin{align*}
	p \max_{1 \leq i \leq p} \mathbb P\left( \left| \sum_{t=1}^T Z''_{i\,t} \right| > {\varepsilon_T''} \right) 
	& \stackrel{(a)}{\leq} {p \over \varepsilon_T''} \max_{1 \leq i \leq p} 
    \mathbb E \left| \sum_{t=1}^T Z''_{i\,t} \right|  
	\leq {p\over \varepsilon_T''} \sum_{t=1}^T \max_{1 \leq i \leq p} \mathbb E | Z''_{i\, t} |
    \nonumber \\
    &\leq {2p \over \varepsilon_T''} \sum_{t=1}^T \max_{1 \leq i \leq p} 
    \mathbb E \left| Z_{i\,t} \mathbbm{1}(|Z_{i\,t}| > b_T) \right| \nonumber \\
	&\stackrel{(b)}{\leq} {2p \over \varepsilon_T''} \sum_{t=1}^T 
    {\max_{1 \leq i \leq p} \mathbb E |Z_{i\,t}|^{r_m} \over b_T^{r_m-1} } \leq {2 p T (2K_m)^{r_m} \over \varepsilon_T'' b_T^{r_m-1}} = {2 K_p (2K_{m})^{r_m} \over (1-\lambda) K \log (T) } ~,
\end{align*}
where $(a)$ follows from Markov's inequality and $(b)$ from the inequality $\mathbb E( | Z \mathbbm{1}(|Z|>b) | ) \leq \mathbb E(|Z|^r)/b^{r-1}$ for a random variable $Z$ with finite $r$-th moment and positive constant $b$.
The claim follows after picking $\lambda= 8 \sqrt{2} / 12$ and noticing that $K = 12 \sigma$ satisfies \eqref{eqn:k1cond}.
\end{proof}







\singlespace
\begin{small}
\bibliography{learning}
\bibliographystyle{natbib}
\end{small} 

\end{document}

%% file: fourth_paragraph.tex
First, this work establishes prediction performance guarantees for 
empirical risk minimization/least squares estimation with dependent data in a large-dimensional setting.
These results allow us to determine under which conditions least squares estimation is a 
reliable estimation strategy in a large-dimensional setup and 
to appraise more precisely the gains of estimation methodologies specifically designed for such a setup.
It is important to acknowledge that estimation methodologies designed for large-dimensional settings 
(for instance, LASSO) typically achieve substantially better performance guarantees than the ones obtained here. 
However, these gains come at the expense of additional assumptions. 
In fact, the performance guarantees obtained here are optimal (up to a logarithmic factor) \citep{Tsybakov:2003}.